\documentclass[11pt]{article}
\usepackage{amsmath,amssymb, fullpage, times}
\usepackage{amsmath, fullpage, times}
\usepackage{algorithmic}
\usepackage{algorithm}
\usepackage{lmodern}

\usepackage{setspace}
\usepackage{enumerate}

\usepackage{todonotes}

\usepackage{xspace} 
\usepackage{epsfig}

\title{$O(\log{\log{rank}})$ \cor \\
 for the \\ 
Matroid Secretary Problem}

\author{Oded Lachish\thanks{Birkbeck, University of London, London, UK. Email:
\hbox{oded@dcs.bbk.ac.uk}}}
\date{}

\date{}

\newcommand{\ignore}[1]{}

\newtheorem{theorem}{Theorem}
\newtheorem{proposition}[theorem]{Proposition}
\newtheorem{lemma}[theorem]{Lemma}
\newtheorem{corollary}[theorem]{Corollary}

\newtheorem{prob}[theorem]{Problem}

\newtheorem{obs}[theorem]{Observation}

\newtheorem{assump}[theorem]{Assumption}

\newtheorem{definition}[theorem]{Definition}

\newcommand{\qed}{\rule{2mm}{2mm}}
\newenvironment{proof}{\par\noindent{\bf Proof.}\quad}{\hfill  $\qed$}

\newcommand{\Prob}[1]{\textrm{prob}\left(#1\right)}

\newcommand{\MSP}{Matroid Secretary Problem}
\newcommand{\CSP}{Classical Secretary Problem}

\newcommand{\SSPI}{Single Sample Prophet Inequality}

\newcommand{\HN}{H}

\newcommand{\PN}{P}

\newcommand{\RN}{R}
\newcommand{\SN}{S}

\newcommand{\ZZ}{\mathbb{Z}}
\newcommand{\NN}{\mathbb{N}}
\newcommand{\NNP}{\NN^{+}}

\newcommand{\OC}{Known-Cardinality}
\newcommand{\OO}{Order-Oblivious}
\newcommand{\MU}{Matroid-Unknown}
\newcommand{\MN}{Matroid-Known}

\newcommand{\medN}{\textrm{med}}
\newcommand{\med}[1]{\medN\left(#1\right)}
\newcommand{\pseudoMid}{\beta(n,1/2)}

\newcommand{\cor}{competitive-ratio}

\newcommand{\val}[1]{\textrm{val}(#1)}

\newcommand{\elementN}{e}

\newcommand{\UU}{U}
\newcommand{\II}{\mathcal I}

\newcommand{\rankN}{rank}
\newcommand{\rank}[1]{\textrm{rank}\left(#1\right)}

\newcommand{\OPTN}{\textrm{OPT}}
\newcommand{\OPT}[1]{\textrm{OPT}\left({#1}\right)}
\newcommand{\LOPTN}{\textrm{LOPT}}
\newcommand{\LOPT}[1]{\textrm{LOPT}\left(#1\right)}

\newcommand{\SpanN}{Closure}
\newcommand{\SpanNA}{\textrm{Cl}}
\newcommand{\Span}[1]{\SpanNA\left(#1\right)}

\newcommand{\Bucket}[2]{\BucketN^{#1}_{#2}}
\newcommand{\BucketN}{B}
\newcommand{\bucketN}{bucket}
\newcommand{\BucketU}[1]{\BucketN_{#1}}

\newcommand{\growthCnst}{2} 

\newcommand{\StageOne}{Gathering stage}
\newcommand{\StageTwo}{Preprocessing stage}
\newcommand{\StageThree}{Selection stage}
\newcommand{\StageTwoWS}{Preprocessing}
\newcommand{\StageThreeWS}{Selection}

\newcommand{\PAlg}{Gap}
\newcommand{\SA}{Simple Algorithm}

\newcommand{\MA}{Main Algorithm}
\newcommand{\TA}{Threshold Algorithm}
\newcommand{\GAPA}{\PAlg\ Algorithm}

\newcommand{\FirstN}{F}
\newcommand{\LastN}{U\setminus F}

\newcommand{\CT}{critical tuple}

\newcommand{\BlockN}{\textrm{Block}}
\newcommand{\Block}[1]{\textrm{Block}(#1)}
\newcommand{\GdN}{\textrm{Good}}
\newcommand{\Gd}[1]{\GdN({#1})}
\newcommand{\BdN}{\textrm{Bad}}
\newcommand{\Bd}[1]{\BdN({#1})}

\newcommand{\BLOCK}{BLOCK}

\newcommand{\uncovN}{\textrm{uncov}}
\newcommand{\uncov}[2]{\textrm{uncov}\left({#1},{#2}\right)}

\newcommand{\indepN}{\textrm{loss}}
\newcommand{\indep}[2]{\textrm{loss}\left({#1},{#2}\right)}

\newcommand{\TradeoffProb}{\frac{1}{4}}
\newcommand{\maxWeightPower}{19} 
\newcommand{\maxWeightPowerMT}{16} 

\newcommand{\KKSet}{K} 

\newcommand{\KSet}{K}
\newcommand{\KSetStar}{\KSet'}

\newcommand{\TalagrandMultCnst}{4}
\newcommand{\TalagrandMultCnstTT}{8}

\newcommand{\TalagrandMultCnstTTPAzzumaDevMultCnstTT}{2^4}

\newcommand{\TalagrandExpAdditiveCnst}{2.1}
\newcommand{\TalagrandExpAdditiveCnstPO}{2.75}

\newcommand{\HSetFirst}{H^{\FirstN}}
\newcommand{\HSetLast}{H^{\LastN}}
\newcommand{\numberOfIter}{4} 

\newcommand{\RVXN}{X}
\newcommand{\RVZN}{Z}

\newcommand{\AzzumaDevMultCnst}{4} 
\newcommand{\AzzumaDevMultCnstTT}{8}
\newcommand{\AzzumaMultCnst}{}
\newcommand{\AzzumaMultCnstTT}{2} 

\newcommand{\SuperSplitMultCnst}{2^7} 

\newcommand{\simpleUnionPower}{\frac{1}{4}} 
\newcommand{\simpleUnionPowerTT}{\frac{1}{2}} 
\newcommand{\rZZDiff}{\frac{3}{4}}
\newcommand{\AzzumaPowerCnst}{\frac{1}{2}}

\newcommand{\singleUnionPower}{\frac{1}{6}}
\newcommand{\singleUnionPowerTT}{\frac{1}{3}}
\newcommand{\singleUnionPowerPHalf}{\frac{2}{3}} 
\newcommand{\singleUnionPowerTTSuperPwr}{\frac{1}{4}} 

\newcommand{\maxWeightPowerTSUPTTSP}{4.75} 
\newcommand{\maxWeightPowerMTTSUPTTSP}{4} 

\newcommand{\simpleGapRankExp}{\frac{3}{4}} 

\newcommand{\manageablePowerDenom}{9}
\newcommand{\manageablePowerEnum}{8}
\newcommand{\manageablePower}{\frac{\manageablePowerEnum}{\manageablePowerDenom}}
\newcommand{\manageablePowerInv}{\frac{\manageablePowerDenom}{\manageablePowerEnum}}

\newcommand{\unionBoundGapPower}{\frac{2}{3}} 
\newcommand{\unionBoundGapPowerThird}{\frac{4}{9}} 
\newcommand{\unionBoundGapPowerSixth}{\frac{1}{3}} 

\newcommand{\concentrationRangeMultCnst}{2^{-5}} 

\newcommand{\manageableGdCnst}{\frac{1}{8}} 

\newcommand{\Super}{Super} 

\newcommand{\PSuperCnst}{-4} 
\newcommand{\PSuperCnstMSeven}{-12} 
\newcommand{\PSuperCnstTSP}{-2.75} 

\newcommand{\SuperPwrDenom}{4}
\newcommand{\SuperPwrEnum}{3}
\newcommand{\SuperPwr}{\frac{\SuperPwrEnum}{\SuperPwrDenom}} 

\newcommand{\superPortionDenom}{9}
\newcommand{\superPortionEnum}{8}
\newcommand{\superPortion}{\frac{\superPortionEnum}{\superPortionDenom}} 

\newcommand{\Valuable}{Valuable}

\newcommand{\valuablePortionDenom}{4}

\newcommand{\valuablePortion}{\frac{3}{4}} 
\newcommand{\valuablePortionInv}{\frac{4}{3}} 

\newcommand{\PValuableIndexLBAdd}{+3} 
\newcommand{\PValuableIndexLBAddMPO}{-1} 

\newcommand{\MValuableCnst}{-1}
\newcommand{\MValuableCnstTVP}{-0.75} 

\newcommand{\valuableRangeMult}{\valuablePortionInv}

\newcommand{\ValuablePwrDenom}{4}
\newcommand{\ValuablePwrEnum}{3}
\newcommand{\ValuablePwr}{\frac{\ValuablePwrEnum}{\ValuablePwrDenom}}

\newcommand{\JSet}{J}

\newcommand{\IncMultCnst}{32}
\newcommand{\IncMultCnstHalf}{16}
\newcommand{\IncMultCnstHalfMO}{15} 

\newcommand{\NLDropCnst}{3}

\newcommand{\NLMap}[1]{%
\ifx&#1&%
{w}%
\else{w(#1)}%
\fi%
}

\newcommand{\TeleCnst}{6}

\newcommand{\SSDropCnst}{18}

\newcommand{\USet}{K} 

\newcommand{\mFam}{\mathcal{H}}
\newcommand{\mFamCardUBCnst}{8}

\newcommand{\Cf}{Critical family}
\newcommand{\cf}{critical family}

\newcommand{\negligible}{negligible}
\newcommand{\splittable}{splittable}

\newcommand{\burned}{burnt}
\newcommand{\usefulN}{useful}

\newcommand{\CTreeN}{critical-tree} 
\newcommand{\CTree}{\mathcal{T}} 
\newcommand{\manageable}{manageable}

\newcommand{\QS}{\mathcal{Q}}

\newcommand{\preDepthMult}{2^7} 

\newcommand{\doNothingMult}{\frac{3}{4}} 
\newcommand{\doNothingMultInv}{\frac{4}{3}} 
\newcommand{\UNSMultConst}{\frac{7}{8}} 

\newcommand{\usefulLBCnst}{\frac{1}{32}} 
\newcommand{\splitLBCnst}{\frac{1}{32}} 
\newcommand{\splitLBCnstOM}{\frac{31}{32}} 
\newcommand{\splitUBCnst}{\frac{15}{16}} 

\newcommand{\PartitionCardCnst}{8} 

\newcommand{\portionOfNegligible}{\frac{1}{4}}
\newcommand{\portionOfBurned}{\frac{1}{8}}
\newcommand{\portionOfBurnedDepthOnePre}{2^{-8}} 
\newcommand{\portionOfBurnedDepthOne}{2^{-3}}

\newcommand{\portionOfUseful}{\frac{1}{2}}
\newcommand{\portionOfUsefulPostStrng}{\frac{1}{36}}

\newcommand{\UsefulDenumCnst}{2^{11}}
\newcommand{\PreTupleCnst}{2^{10}} 

\newcommand{\PreSetCnst}{2^{11}} 

\newcommand{\manageableMultCnst}{\frac{1}{2}}

\newcommand{\GTSetN}{M}
\newcommand{\GTSet}[1]{\GTSetN\left(#1\right)}
\newcommand{\GTSetH}[1]{\GTSetN_H\left(#1\right)}

\newcommand{\preLOPTLBDiv}{8}

\newcommand{\PISHeavy}{\gamma} 

\newcommand{\singleSuperRankLB}{9}

\newcommand{\slackMultCnst}{2^{9}} 

\newcommand{\slackPwrOM}{\frac{7}{8}} 
\newcommand{\slackPwr}{\frac{1}{8}} 
\newcommand{\slackPwrTVal}{\frac{3}{32}} 

\newcommand{\SlackCnst}{2^{30}}

\newcommand{\SlackCnstTTF}{2^{36}}

\newcommand{\slackLogExpCnst}{30} 

\newcommand{\valuableRankSplitMult}{2^{9}} %
\newcommand{\SuperSingleBucketFirstPortion}{\frac{1}{4}}
\newcommand{\SuperSingleBucketFirstPortionInv}{4}

\newcommand{\LastSetCnst}{2^{13}} 
\newcommand{\LastSetCnstTT}{2^{15}} 

\newcommand{\KOnePortion}{\frac{1}{12}}
\newcommand{\KTwoPortion}{\frac{1}{12}}
\newcommand{\KPortionTT}{\frac{1}{6}}
\newcommand{\FirstPortionInUU}{\frac{2}{9}}

\newcommand{\singleFirstLastRankRatio}{3} 

\newcommand{\tupleMainMultCnst}{2^{-5}} 
\newcommand{\SetMainMultCnst}{2^{-3}} 

\begin{document}
\maketitle

\thispagestyle{empty}

\begin{abstract}
\sloppypar
In the  \textit{\MSP} (MSP), the elements of the ground set of a Matroid are revealed on-line one by one, each together with its value.
An algorithm for the \MSP\  is \textit{\MU} if,
at every stage of its execution: (i) it only knows the elements that have been revealed so far and their values, and (ii) it has access to 
 an oracle for testing whether or not any subset of the elements that have been revealed so far is an independent set.
An algorithm is \textit{\OC} if, in addition to (i) and (ii), it  also initially knows the cardinality of the ground set of the Matroid. 

We present here a \OC\ and \textit{\OO} algorithm that, with constant probability, selects an independent set of elements, whose value is at least the optimal value divided by $O(\log{\log{\rho}})$, where $\rho$ is the rank of the Matroid;
that is, the algorithm has a \textit{\cor} of $O(\log{\log{\rho}})$.
The best previous results for a \OC\ algorithm are a \cor\ of $O(\log{\rho})$, by Babaioff \textit{et al.} (2007), and
a \cor\ of $O(\sqrt{\log{\rho}})$, by Chakraborty and Lachish (2012).

In many non-trivial cases the algorithm we present has a \cor\ that is better than the $O(\log{\log{\rho}})$.
The cases in which it fails to do so are easily characterized.
Understanding these cases may lead to improved algorithms for the problem or, conversely, to non-trivial lower bounds.

\end{abstract}

\section{Introduction}

The \textit{\MSP} is a generalization of the \textit{\CSP}, whose origins seem to still be a source of dispute.
One of the first papers on the subject~\cite{dynkin}, by Dynkin, dates back to 1963.
Lindley~\cite{lindley} and Dynkin~\cite{dynkin} each presented an algorithm that achieves a \textit{\cor} of $e$, which is the best possible.
See~\cite{Freeman} for more information about results preceding 1983.

In 2007, Babaioff \textit{et al.}~\cite{babaioff1} established a connection between the \MSP\ and {\it mechanism design}.
This is probably the cause of an increase of interest in generalizations of the \textit{\CSP } and specifically the \MSP.

In the \MSP, we are given a Matroid $\{\UU,\II\}$ and a value function assigning non-negative values to the Matroid elements.
The elements of the Matroid are revealed in an on-line fashion according to an unknown order selected uniformly at random.
The value of each element is unknown until it is revealed.
Immediately after each element is revealed, if the element together with the elements already selected does not form an independent set,  then that element cannot be selected; however, if it does, then an irrevocable decision must be made whether or not to select the element.
That is, if the element is selected, it will stay selected until the end of the process and likewise if it is not.
The goal is to design an algorithm for this problem wit ha small  \cor, that is the ratio between the maximum sum of values of an independent set and the expected sum of values of the independent set returned by the algorithm. 

An algorithm for the \textit{\MSP} (MSP) is called \textit{\MU} if,
at every stage of its execution, it only knows (i) the elements that have been revealed so far and their values and (ii)
 an oracle for testing whether or not a subset the elements that have been revealed so far forms an independent set.
An algorithm is called \textit{\OC} if it knows (i), (ii) and also knows from the start
 the cardinality $n$ of the ground set of the Matroid. 
An algorithm is called \textit{\MN}, if it knows, from the start, everything about the Matroid except for the values of the elements. 
These, as mentioned above, are revealed to the algorithm as each element is revealed.

\paragraph{Related Work}
Our work follows the path initiated by Babaioff \textit{et al.} in~\cite{babaioff1}.
There they formalized the \MSP\ and presented a \OC\ algorithm with a \cor\ of $\log{\rho}$. 
This line of work was continued in~\cite{improved}, where an algorithm with a  \cor\ of $O(\sqrt{\log{\rho}})$ was presented.
In Babaioff \textit{et al.}~\cite{babaioff1} (2007), it was conjectured that a constant \cor\ is achievable.
The best known result for a \textit{\MU} algorithm, implied by the works of Gharan and Vondr\'{a}ck~\cite{Variants} and Chakraborty and Lachish~\cite{improved} (2012): for every fixed $\epsilon > 0$, there exists a \MU\ algorithm with a \cor\ of  $O(\epsilon^{-1}(\sqrt{\log{\rho}})\log^{1+\epsilon}{n})$. 
Gharan and Vondr\'{a}ck showed that a lower bound of $\Omega(\frac{\log{n}}{\log\log{n}})$ on the \cor\ holds in this case. 

Another line of work towards resolving the \MSP\ is the study of the Secretary Problem for specific families of Matroids.
Most of the results of this type are for \MN\ algorithms and all achieve a constant \cor. 
Among the specific families of Matroids studied are \textit{Graphic Matroids} \cite{babaioff1}, \textit{Uniform/Partition Matroids} \cite{Knapsack,Uniform2}, \textit{Transversal Matroids} \cite{transversal1,transversal2}, \textit{Regular and Decomposable Matroids}~\cite{Dinitz} and \textit{Laminar Matroids}~\cite{laminar}.
For surveys that also include other variants of the \MSP~see~\cite{Assignment,Advances,dinitz2013recent}.

There are also results for other generalizations of the \CSP, including the
\textit{Knapsack Secretary Problem}~\cite{Knapsack}, 
\textit{Secretary Problems with Convex Costs}~\cite{convex}, \textit{Sub-modular Secretary Problems}~\cite{submodular,Gupta,feldman} and 
\textit{Secretary problems via linear programming}~\cite{LP}.

\sloppypar
\paragraph{Main result}
We present here a \OC\ algorithm with a \cor\ of $O(\log{\log{\rho}})$. 
The algorithm is also \OO\ as defined by Azar \textit{et al.}~\cite{ProphetLI}). Definition~\ref{def:OO} is a citation of their definition of an \OO\ algorithm for the \MSP.
According to~\cite{Variants}, this implies that, for every fixed $\epsilon > 0$, there exists a \MU\ algorithm with a \cor\ of  $O(\epsilon^{-1}(\log{\log{\rho}})\log^{1+\epsilon}{n})$. 
Our algorithm is also \OO\ as in Definition 1 of~\cite{ProphetLI}, and hence, by Theorem 1 of~\cite{ProphetLI}, this would imply that there exists a \SSPI\ for Matroids with a \cor\ of $O(\log{\log{\rho}})$.

In many non-trivial cases the algorithm we present has a \cor\ that is better than the $O(\log{\log{\rho}})$.
The cases in which it fails to do so are characterized.
Understanding these cases may lead to improved algorithms for the problem or, conversely, to non-trivial lower bounds.

\paragraph{High level description of result and its relation to previous work.}
As in~\cite{babaioff1} and~\cite{improved}, here we also partition the elements into sets which we call \textit{buckets}.
This is done by rounding down the value of each element to the largest possible power of two and then, 
for every power of two, 
defining a bucket to be the set of all elements with that value.
Obviously, the only impact this has on the order of the \cor\ achieved is a constant factor of at most $2$.

We call our algorithm the \MA. It has three consecutive stages: \textit{\StageOne}, \textit{\StageTwo} and \textit{\StageThree}.
In the \StageOne\ it waits, without selecting any elements, until about half of the elements of the matroid are revealed.
The set $\FirstN$ that consists of all the elements revealed during the \StageOne\ is the input to the \StageTwo.
In the \StageTwo, on out of the following three types of output is computed: 
(i) a non negative value, (ii) a set of bucket indices, or (iii) a \CT.
Given the output of the \StageTwo, before any element is revealed
the \MA\ chooses one of the following algorithms:
 the \textit{\TA},
the \textit{\SA} or the \textit{\GAPA}.
Then, after each one of the remaining elements is revealed,
the decision whether to select the element is made by
the chosen algorithm using the input received from the \StageTwo\ and the set of all the elements already revealed.
Once all the elements have been revealed the set of selected elements is returned.

The \TA\ is chosen when the output to the \StageTwo\ is a non-negative value,
which happens with probability half regardless of the contents of the set $\FirstN$.
 Given this input, the \TA, as in the algorithm for the \CSP, selects only the first element that has at least the given value.
The \SA\ is chosen when the output of \StageTwo\ is a set of bucket indices.
The \SA\, selects an element if it is in one of the buckets determined by the set of indices and if it is independent of all previously selected elements.
This specific algorithm was also used in~\cite{improved}.

The \GAPA\ is chosen when the output of \StageThree\ is a \CT, which we define further on.
The \GAPA\ works as follows:
 every element revealed is required to have one of a specific set of values and satisfy two conditions in order to be selected:
 it satisfies the first condition if it is in the closure of a specific subset of elements of $\FirstN$; it satisfies the second condition if it is not in the closure of the union of the set of elements already selected and a  specific subset of elements of $\FirstN$ (which is different than the one used in the first condition).

The proof that the \MA\ achieves the claimed \cor\ consists of the following parts:
a guarantee on the output of the  \SA\ as a function of the input and $\LastN$, where $\UU$ is the ground set of the matroid; 
a guarantee on the output of the  \GAPA\ as a function of the input and $\LastN$; 
a combination of a new structural result for matroids and probabilistic inequalities that imply that if the matroid does not have an element with a large value, then it is possible to compute an input for either the \SA\ or the \GAPA\ that, with high probability,  ensures that the output set has a high value.
This guarantees the claimed \cor , since the case when the matroid has 
an element with a large value is dealt with by the \TA.

\paragraph{The paper is organized as follows:}
Section~\ref{sec:preliminaries} contains the preliminaries; Section~\ref{sec:overview} presents \MA; 
Section~\ref{sec:GAPA} is devoted to the \SA\ and the \GAPA;
Section~\ref{sec:concentrations} contains the required concentrations;
the structural trade-off result is proved in
Section~\ref{sec:structure};
 the main result appears in Section~\ref{sec:MainResult};
 and in Section~\ref{sec:Discussion} we characterize the cases in which the algorithm performs exactly as guaranteed and give non-trivial example in which the algorithm performs better than the guaranteed \cor.

\section{Preliminaries}\label{sec:preliminaries}
All logarithms are to the base $\growthCnst$.
We use $\ZZ$ to denote the set of all integers, $\NN$ to denote the non-negative integers and 
$\NNP$ to denote the positive integers.
We use $[\alpha]$ to denote  $\{1,2,\dots,\lfloor\alpha\rfloor\}$ for any non-negative 
 real $\alpha$.
 We use $[\alpha,\beta]$ to denote $\{i\in \ZZ \mid \alpha \leq i \leq\beta\}$ and $(\alpha,\beta]$ to denote $\{i\in \ZZ \mid \alpha < i \leq\beta\}$, and so on.
 We use $\med{f}$ to denote the \textit{median} of a function $f$ from a finite set to the non-negative reals. If there are two possible values for $\med{f}$ the smaller one is chosen. 
 
We define $\pseudoMid$ to be a random variable whose value is the number of successes in $n$ independent probability $1/2$ Bernoulli trials.
\begin{obs}~\label{obs:binomial}
Let $A = \{a_1,a_2,\dots,a_n\}$ and $W = \pseudoMid$; let
 $\pi:[n]\longrightarrow [n]$ be a permutation selected uniformly at random, and let $D = \{a_{\pi(i)} \mid i \in [W]\}$.
For every $i\in [n]$, we have that $a_i\in D$ independently with probability $1/2$.
\end{obs}
\begin{proof}
To prove the proposition we only need to show that for every $C\subseteq A$, we have
$D = C$ with probability $2^{-n}$.
Fix $C$. 
There are $\binom{n}{|C|}$ subsets of $A$ of size $|C|$.
$D$ is equally likely to be one of these subsets.
Hence, the probability that $|D|=|C|$ is $\binom{n}{|C|}\cdot 2^{-n}$ and therefore the probability that $D=C$ is  $\binom{n}{|C|}\cdot 2^{-n}/\binom{n}{|C|} = 2^{-n}$.
\end{proof}

\subsection{Matroid definitions, notations and preliminary results}
\begin{definition}{\textbf{[Matroid]}}~\label{def:matroid}
A \textbf{matroid} is an ordered pair $M=(\UU,\II)$, where $\UU$ is a set of \textbf{elements}, called the \textit{ground set}, and $\II$ is a family of subsets of $\UU$ that satisfies the following:
\begin{itemize} 
\item If $I\in \II$ and $I'\subset I$, then $I'\in \II$
\item If $I, I'\in \II$ and $|I'| < |I|$, then there exists $e\in I\setminus I'$ such that $I'\cup\{e\}\in\II$. 
\end{itemize}
The sets in $\II$ are called \textbf{independent} sets and a maximal independent set is called a \textbf{basis}.
\end{definition}

A \textit{value function} over a Matroid $M=(\UU,\II)$ is a mapping from the elements of $\UU$ to the non-negative reals. 
Since we deal with a fixed Matroid and value function,
we will always use $M=(\UU,\II)$ for the Matroid.
We set $n = |\UU|$ and, for every $e\in \UU$, we denote its value by $\val{e}$.
\begin{definition}{\textbf{[\rankN\ and \SpanN]}}\label{def:RankSpan}
For every $\SN \subseteq \UU$, let 
\begin{itemize}
\item
$\rank{\SN} = \max\{|\SN'|\mid \SN' \in \II \mbox{~and~} \SN'\subseteq \SN\}$ and
\item
$\Span{\SN} = \left\{\elementN\in \UU\mid \rank{\SN\cup \{\elementN\}} = \rank{\SN} \right\}$.
\end{itemize}
\end{definition}

The following proposition captures a number of standard properties of Matroids; the proofs can be found in~\cite{MatroidBook}. We shall only prove the last assertion.
\begin{proposition}~\label{prop:Matroid}
Let $\SN_1,\SN_2,\SN_3$ be subsets of $\UU$ and $\elementN\in \UU$ then
\begin{enumerate}
\item\label{item:CardRank} $\rank{\SN_1} \leq |\SN_1|$, where equality holds if and only if $\SN_1$ is an independent set,
\item\label{item:MatroidClContained} if $\SN_1 \subseteq \SN_2$ or $\SN_1 \subseteq \Span{\SN_2}$, then $\SN_1 \subseteq \Span{\SN_1} \subseteq \Span{\SN_2}$ and $\rank{\SN_1} \leq \rank{\SN_2}$,
\item\label{item:MatroidPlusOne} if $\elementN\not\in \Span{\SN_1}$, then 
$\rank{\SN_1\cup\{\elementN\}} = \rank{\SN_1} + 1$,
\item\label{item:MatroidUnion} $\rank{\SN_1 \cup \SN_2} \leq \rank{\SN_1} + \rank{\SN_2}$,
\item\label{item:MatroidUnionSpan} $\rank{\SN_1 \cup \SN_2} \leq \rank{\SN_1} + \rank{\SN_2\setminus \Span{\SN_1}}$,
 and
\item\label{item:MatroidExchange} suppose that $\SN_1$ is minimal such that  $\elementN\in\Span{\SN_1 \cup \SN_2}$, but $\elementN\not\in\Span{(\SN_1 \cup \SN_2) \setminus \{\elementN^*\}}$, for every 
$\elementN^* \in \SN_1$,
then 
$\elementN^*\in \Span{\{\elementN\}\cup ((\SN_1 \cup \SN_2)\setminus \{\elementN^*\})  }$,
for every $\elementN^*\in \SN_1$. 
\end{enumerate}
\end{proposition} 
\begin{proof}
We prove Item~\ref{item:MatroidExchange}. 
The rest of the items are standard properties of Matroids.
\sloppypar{
Let $\elementN^* \in \SN_1$.
By Item~\ref{item:MatroidPlusOne}, 
$\rank{\{\elementN\}\cup ((\SN_1 \cup \SN_2)\setminus \{\elementN^*\})}$ is equal to
$\rank{\SN_1 \cup \SN_2}$ which is equal to
$\rank{\{\elementN\}\cup \SN_1 \cup \SN_2}$ which in turn is equal to
$\rank{\{\elementN\}\cup ((\SN_1 \cup \SN_2)\setminus \{\elementN^*\}) \cup \{\elementN^*\}}$.
Thus, again by Item~\ref{item:MatroidPlusOne},
this implies that 
$\elementN^* \in \Span{(\{\elementN\}\cup ((\SN_1 \cup \SN_2)\setminus \{\elementN^*\})}$.}
\end{proof}


\begin{assump}~\label{ass:rank}
$\val{\elementN}= 0$, for every $\elementN\in \UU$ such that $\rank{\{\elementN\}}=0$.
For every $\elementN\in \UU$ such that $\val{\elementN}> 0$, there exists $i\in \ZZ$ such that $\val{\elementN} = \growthCnst^i$.
\end{assump}

In the worst case, the implication of this assumption is an increase in the competitive ratio by a multiplicative factor that does not exceed $\growthCnst$, compared with the competitive ratio we could achieve without this assumption.

\begin{definition}{\textbf{[Buckets]}}~\label{def:Bucket}
For every $i\in \ZZ$, the $i$'th bucket is $\BucketU{i} = \{e \in \UU \mid \val{e}  = 2^i\}.$ We also use the following notation for every $\SN \subseteq \UU$ and $J\subset \ZZ$:
\begin{itemize}
\item $\Bucket{\SN}{i} = \BucketU{i}\cap \SN$, 
\item $\BucketU{J} = \bigcup_{i\in J}\BucketU{i}$ and  
\item $\Bucket{\SN}{J} = \bigcup_{i\in J}\Bucket{\SN}{i}$.
\end{itemize}
\end{definition}

\begin{definition}{\textbf{[\OPTN]}}~\label{def:OPT}
For every $\SN \subseteq \UU$, let 
  $\OPT{\SN} = 
\max\left\{\sum_{\elementN\in \SN'}\val{\elementN}\Big|~ \SN'\subseteq\SN \mbox{ and } \SN' \in \II \right\}.
$
\end{definition}
We note that if $\SN$ is independent, then $\OPT{\SN} = \sum_{\elementN\in \SN}\val{\elementN}$.

\begin{obs}\label{obs:OPT}
For every independent $\SN\subseteq \UU$, $\OPT{\SN} = \sum_{i\in \ZZ}2^i\cdot \rank{\Bucket{\SN}{i}}$.
\end{obs}

\begin{definition}{\textbf{[\LOPTN]}}~\label{def:LOPT}
For every $\SN \subseteq \UU$, we define 
$\LOPT{\SN} = \sum_{i\in \ZZ}\growthCnst^i\cdot \rank{\Bucket{\SN}{i}}$.
\end{definition}

\begin{obs}~\label{obs:LOPT}
For every $\SN\subseteq \UU$ and  $J_1,J_2\subseteq \ZZ$,
\begin{enumerate}
\item\label{item:LOPT1} $\LOPT{\SN} \geq \OPT{\SN}$,
\item\label{item:LOPT1H}
 $\LOPT{\Bucket{\SN}{J_1}} = \sum_{i\in J_1}\growthCnst^i\cdot\rank{\Bucket{\SN}{i}}$ and
\item\label{item:LOPT2} if $J_1\cap J_2 = \emptyset$, then 
$\LOPT{\Bucket{\SN}{J_1\cup J_2}} = 
\LOPT{\Bucket{\SN}{J_1}} +
\LOPT{\Bucket{\SN}{J_2}}.$
\end{enumerate}
\end{obs}

\subsection{\MSP}

\begin{definition}{\textbf{[\cor]}}~\label{def:cor}
Given a Matroid $\mathcal{M} = (\UU,\II)$, 
the \textit{\cor} of an algorithm that selects an independent set $\PN\subseteq \UU$ is the ratio of $\OPT{\UU}$ to the expected value of $\OPT{\PN}$.
\end{definition}

\begin{prob}{\textbf{[\OC\ Matroid Secretary Problem]}}~\label{prob:MSP}
The elements of the Matroid $M=(\UU,\II)$
are revealed in random order in an on-line fashion.
The cardinality of $\UU$ is known in advance, but 
every element and its value are unknown until revealed.
The only access to the structure of the Matroid is via an oracle that, upon receiving a query in the form of a subset of elements already revealed, answers whether the subset is independent or not.
An element can be selected only after it is revealed and before the next element is revealed, and then only provided the set of selected elements remains independent at all times.
Once an element is selected it remains selected.
The goal is to design an algorithm that maximizes the expected value of $\OPT{\PN}$, i.e.,
achieves as small a \textbf{\cor} as possible.
\end{prob}

\begin{definition}\label{def:OO}\textbf{(Definition 1 in~\cite{ProphetLI})}. 
We say that an algorithm $\mathcal{S}$ for the secretary problem (together with its corresponding analysis)
is \textbf{order-oblivious} if, on a randomly ordered input vector $(v_{i_1},\dots,v_{i_n})$:
\begin{enumerate}
\item (algorithm)  $\mathcal{S}$ sets a (possibly random) number $k$, observes without accepting the first $k$ values 
$S =\{v_{i_1},\dots,v_{i_k}\}$, and uses information from $S$ to choose elements from $V = \{v_{i_{k+1}},\dots,v_{i_n}\}$.
\item (analysis) $S$ maintains its competitive ratio even if the elements from $V$ are revealed in any (possibly
adversarial) order. In other words, the analysis does not fully exploit the randomness in the arrival
of elements, it just requires that the elements from $S$ arrive before the elements of $V$, and that the
elements of $S$ are the first $k$ items in a random permutation of values.
\end{enumerate}
\end{definition}
\section{The \MA}\label{sec:overview}
The input to the \MA\ is the number of indices $n$ in a randomly ordered input vector $(\elementN_1,\dots,\elementN_n)$, where $\{\elementN_1,\dots,\elementN_n\}$ are the elements of the ground set of the matroid.
These are revealed to the \MA\ one by one in an on-line fashion in the increasing order of their indices.
The \MA\ executes the following three stages:
\begin{enumerate}
\item 
\textbf{\StageOne.} 
Let $W=\pseudoMid$.
Wait until $W$ elements are revealed without selecting any. Let $\FirstN$ be the set of all these elements.
\item 
\textbf{\StageTwo.}
Given only $\FirstN$, before any item of $\LastN$ is revealed, one of the following three types of output is computed:
(i) a non-negative value, (ii)
a set of bucket indices, or (iii) a \CT\ which is defined in Subsection~\ref{subsec:GAPA}.
\item\label{StageSelection}
\textbf{\StageThree.}
One out of three algorithms is chosen and used
in order to decide which elements from $\LastN$ to select, when they are revealed.
If the output of \StageTwo\ is a non-negative value, then the \textit{\TA} is chosen,
if it is a set of bucket indices, then the \textit{\SA} is chosen and if it is a \CT, then the \textit{\GAPA} is chosen.
\end{enumerate}

With probability $\frac{1}{2}$, regardless of $\FirstN$, the output of the \StageTwo\ is the largest value of the elements of $\FirstN$.
The \TA, which is used in this case, selects the first revealed element of $\LastN$ that has a value at least as large as the output of the \StageTwo.
This ensures that if
$\max\{\val{\elementN}\mid \elementN \in \UU \} \geq 2^{-\maxWeightPower}\cdot \OPT{\UU}$,
then the claimed \cor\ is achieved.
So for the rest of the paper we make the following assumption:
\begin{assump}\label{ass:MaxVal}
$\max\{\val{\elementN}\mid \elementN \in \UU \} < 2^{-\maxWeightPower}\cdot \OPT{\UU}$.
\end{assump}

The paper proceeds as follows:
in Subsection~\ref{subsec:SA}, we present the \SA\ and formally prove a guarantee on its output;
in Subsection~\ref{subsec:GAPA}, we define \CT, describe the \GAPA\ and formally prove a guarantee on its output;
in Section~\ref{sec:concentrations}, prove the required concentrations;
in Section~\ref{sec:structure},
we prove our structural trade-off result; and in Section~\ref{sec:MainResult}, we prove the main result.

\section{The \SA\ and the \GAPA}\label{sec:GAPA}
In this section we present the pseudo-code for the \SA\ and the \GAPA, and prove the guarantees on the \cor s they achieve.
We start with the \SA, which is also used in~\cite{improved}.

\subsection{The \SA}\label{subsec:SA}
\begin{algorithm}
\caption{\SA}
\label{alg:SA}
{\bf Input: a set $\JSet$ of bucket indices}
\begin{enumerate}
	\item\label{step:SAinitializeP} 
		$\PN\longleftarrow \emptyset$
	\item\label{step:SAreveal} 
		immediately after each element $\elementN\in \LastN$ is revealed, do 
		\begin{enumerate}
  			\item\label{step:SAJ}
  				if $\log{\val{\elementN}}\in \JSet$ do 		
			\begin{enumerate}
  				\item\label{step:SASelect} 
  					if $\elementN\not\in \Span{\PN}$ do 
 					$\PN\longleftarrow \PN\cup \{\elementN\}$
  			\end{enumerate}
  		\end{enumerate}
\end{enumerate}
{\bf Output}: $\PN$
\end{algorithm}

We note that according to Steps~\ref{step:SAJ} and~\ref{step:SASelect}, the output $\PN$ of the \SA\ always satisfies, $\Bucket{\LastN}{\JSet}\subseteq \Span{\PN}$.
Thus, since $\PN \subseteq \Bucket{\LastN}{\JSet}$, the output $\PN$ of the \SA\ always satisfies, $\rank{\PN} = \rank{\Bucket{\LastN}{\JSet}}$.
As a result, for every $j\in \JSet$, we are guaranteed that $\PN$ contains at least $\rank{\Bucket{\LastN}{\JSet}} - \rank{\Bucket{\LastN}{\JSet\setminus \{j\}}}$ elements from $\Bucket{\LastN}{j}$.
We capture this measure using the following definition:
\begin{definition}{\textbf{[\uncovN]}}\label{def:uncov}
$\uncov{\RN}{\SN} =  \rank{\RN\cup\SN} - \rank{\RN},~~$ for every $\RN,\SN \subseteq \UU$. 
\end{definition} 
It is easy to show that
\begin{obs}\label{obs:uncov} 
$\uncov{\RN}{\SN}$ is monotonic decreasing in $\RN$.
\end{obs} 
According to this definition, for every $j\in \JSet$, we are guaranteed that $\PN$ contains at least $\uncov{\Bucket{\LastN}{\JSet\setminus \{j\}}}{\Bucket{\LastN}{j}}$ elements from $\Bucket{\LastN}{j}$. 
We next prove this in a slightly more general setting that is required for the \GAPA.
\begin{lemma}\label{lem:greedy}
Suppose that the input to the \SA\ is a set $\JSet\subset \ZZ$  
and, instead of the elements of $\LastN$, the elements of a set $S\subseteq \UU$ are revealed in an arbitrary order to the \SA. 
Then the \SA\ returns an independent set  $\PN\subseteq S$ such that, for every $j\in \JSet$,
$
\rank{\Bucket{\PN}{j}} \geq 
\uncov{\Bucket{S}{\JSet\setminus\{j\}}}{\Bucket{S}{j}}.
$
\end{lemma}
\begin{proof}
By the same reasoning as described in the beginning of this section, for every $j\in \JSet$, we are guaranteed that $\PN$ contains at least $\uncov{\Bucket{\LastN}{\JSet\setminus \{j\}}}{\Bucket{\LastN}{j}}$ elements from $\Bucket{\LastN}{j}$ and
the result follows.
\end{proof}

We next prove the following guarantee on the output of the \SA, by using the preceding lemma.
\begin{theorem}\label{thm:greedy}
Given a set $\JSet\subset \ZZ$ as input, the \SA\ returns an independent set  $\PN\subseteq \LastN$ such that
$$\OPT{\PN}\geq \sum_{j\in \JSet}2^j\cdot\uncov{\Bucket{\LastN}{\JSet\setminus\{j\}}}{\Bucket{\LastN}{j}}.$$
\end{theorem}
\begin{proof}
By Observation~\ref{obs:OPT},  $\OPT{\PN}$ is at least 
$\sum_{j\in \JSet}2^{j}\cdot\rank{\Bucket{\PN}{j}}$, which is at least
$\sum_{j\in \JSet}2^{j}\cdot\uncov{\Bucket{\LastN}{\JSet\setminus\{j\}}}{\Bucket{\LastN}{j}},
$
by Lemma~\ref{lem:greedy}.
The result follows.
\end{proof}

We note that the above guarantee is not necessarily the best possible. However, it is sufficient for our needs because, as we show later on, with very high probability, for a specific family of sets $\JSet$ and every $j$ in such $\JSet$,  we have that $\uncov{\Bucket{\LastN}{\JSet\setminus\{j\}}}{\Bucket{\LastN}{j}} \approx \uncov{\Bucket{\FirstN}{\JSet\setminus\{j\}}}{\Bucket{\FirstN}{j}}$. 
Thus, in relevant cases, we can approximate this guarantee using only the elements of $\FirstN$. 

\begin{corollary}\label{cor:SA}
Given a set $\JSet= \{k\}$ as input, the \SA\ returns an independent set  $\PN\subseteq \LastN$ such that
 $\OPT{\PN} \geq 2^k\cdot\rank{\Bucket{\LastN}{k}}$.
\end{corollary}

\subsection{The \GAPA}\label{subsec:GAPA}
The subsection starts with a description of the input to the \GAPA\ and how it works;
afterwards it provides a formal definition of the \GAPA\ and its input and then concludes with a formal proof of the guarantee on the \GAPA's output.

Like the \SA\ the elements of $\LastN$ are revealed to the \GAPA\ one by one in an on-line manner.
The input to the \GAPA\ is a tuple $(\BlockN,\GdN,\BdN)$, called a critical tuple.
$\BlockN$ is a mapping from the integers $\ZZ$ to the power set of the integers,
such that if $\Block{i}$ is not empty then $i\in \Block{i}$.
$\BlockN$ determines from which buckets the \GAPA\ may select elements. 
Specifically, an element $\elementN \in \LastN$ may be selected only if  $\Block{\log{\val{\elementN}}}$ is not empty.
Every pair of not empty sets $\Block{i}$ and $\Block{j}$, where $i\geq j$, are such that either $\Block{i} = \Block{j}$ or $\min{\Block{i}} > \max{\Block{j}}$ and the latter may occur only if $i>j$.
We next formally define the \CT.

\begin{definition}{\textbf{[\CT, \BLOCK]}}\label{def:tuple}
$(\BlockN,\GdN,\BdN)$,
where 
$\GdN$, $\BdN$ and $\BlockN$ are mappings from $\ZZ$ to $2^{\ZZ}$, is a \textbf{\CT} if the following hold for every $i,j\in \ZZ$ such that $i\geq j$ and $\Block{i}$ and $\Block{j}$ are not empty:
\begin{enumerate}
\item\label{item:CTBlock} $i\in \Block{i}$, 
\item\label{item:CTBlocks}
if $i >j$ either $\Block{i} = \Block{j}$ or $\min{\Block{i}} >  \max{\Block{j}}$,
\item\label{item:CTSameBlock}  if $\Block{i} = \Block{j}$, then $\Gd{i} = \Gd{j}$ and  $\Bd{i} = \Bd{j}$, 
\item\label{item:BlockGood} $\Block{i}\cup \Bd{i} \subseteq \Gd{i}$,
\item\label{item:CTLO} 
if $\min{\Block{i}} > \max{\Block{j}}$, then $\Bd{i} \subseteq \Gd{i} \subseteq \Bd{j} \subseteq \Gd{j}$,
\item\label{item:BBGStructure} 
 $\max{\Block{i}} <  \min{\Bd{i}}$. 
\end{enumerate}
We define $\BLOCK = \{i\mid \Block{i}\neq\emptyset\}$.
\end{definition}
For a depiction of the preceding structure see Figure~\ref{fig:CT}.

\begin{figure}[ht]
	\centering
   \includegraphics[width=2.4in]{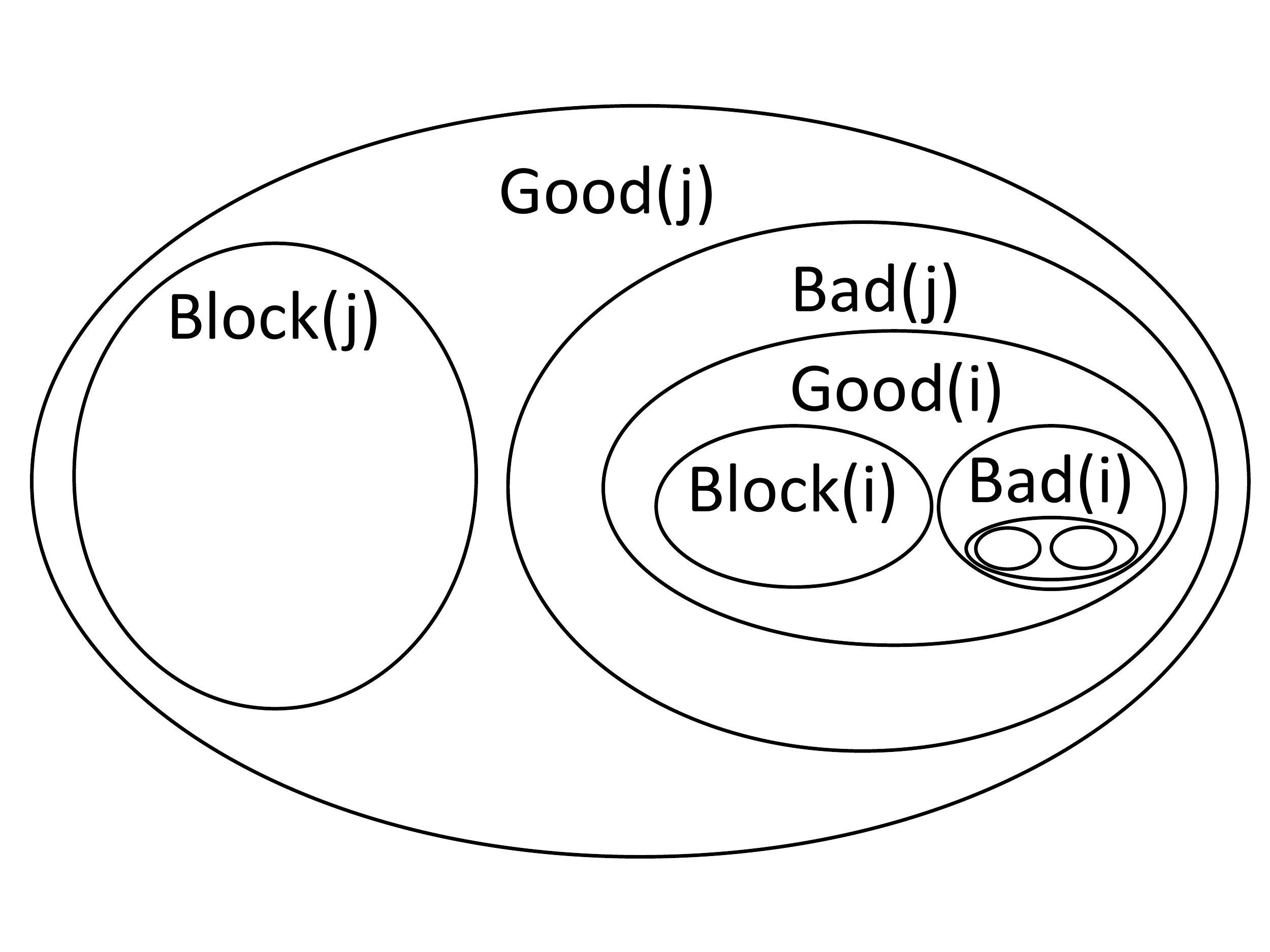}
   \caption{[\CT, where $\min{\Block{i}} > \max{\Block{j}}$]}\label{fig:CT}
\end{figure}
\noindent
The following observation, follow directly from the preceding definition.

\begin{obs}\label{obs:badblock}
If $(\BlockN,\GdN,\BdN)$ is a \textbf{\CT}, then 
\begin{enumerate}
\item\label{item:badblockDisjoint} the sets in $\{\Block{j}\}_{j\in \ZZ}$ are pairwise-disjoint,
\item $\Block{i} \cap \Bd{i} = \emptyset$, for every $i\in \ZZ$,
 and
\item\label{item:badblockGood} for every $i$ and $j$ in $\bigcup_{\ell\in \ZZ} \Block{\ell}$, if $j\not\in\Gd{i}$, then $i > j$ and 
$\Gd{i} \subseteq \Bd{j}$.
\end{enumerate} 
\end{obs}
The mappings $\GdN$ and $\BdN$ are used in order to determine if an element can be selected as follows:
an element $\elementN \in \LastN$ such that
$\log{\val{\elementN}}\in \Block{\log{\val{\elementN}}}$ is selected if it satisfies two conditions:
(i)
$\elementN \in \Span{\Bucket{\FirstN}{\Gd{i}}}$; and (ii)
$\elementN$ is in the closure of the union of $\Bucket{\FirstN}{\Bd{i}}$ and all the previously selected elements.
We next explain why this strategy works.

Clearly, the only elements in $\Bucket{\LastN}{i}$ that do not satisfy condition (i) are those in $\Bucket{\LastN}{j} \setminus \Span{\Bucket{\FirstN}{\Gd{i}}}$.
An essential part of our result is an upper bound on the rank of the set 
$\Bucket{\LastN}{i} \setminus \Span{\Bucket{\FirstN}{\Gd{i}}}$ 
and hence we use the following definition to capture this quantity.
\begin{definition}{\textbf{[\indepN]}}\label{def:Ind}
 For every $\RN,\SN \subseteq \UU$, let
$\indep{\RN}{\SN} = 
 \rank{\SN \setminus \Span{\RN}}.$
\end{definition}
According to this definition and the preceding explanation we are guaranteed
that the rank of the set of elements in $\Bucket{\LastN}{i}$ that satisfy condition (i) is at least
$\rank{\Bucket{\LastN}{i}}-\indep{\Bucket{\FirstN}{\Gd{i}}}{\Bucket{\LastN}{i}}$.

For every $j\in \Block{i}$, let $S_j = \Bucket{\LastN}{j} \cap \Span{\Bucket{\FirstN}{\Gd{j}}}$, that is, the elements of $S_j$ are the elements of $\Bucket{\LastN}{j}$ that satisfy condition (i). 
We will show that such an element satisfies condition (ii) if
it is not in the closure of the union of $\Bucket{\FirstN}{\Bd{i}}$ and \textbf{only} all the elements from $\Bucket{\LastN}{i}$ that were previously selected.
The reason this happens is that, for every $j'>i$, such that $\min{\Block{j'}}> \max{\Block{i}}$ all the element selected from $\Bucket{\LastN}{j'}$, satisfy condition (i) and hence are in $\Span{\Bucket{\FirstN}{\Bd{i}}}$, and for every $j'<i$, such that $\max{\Block{j'}}< \min{\Block{i}}$ the condition (ii) ensures, for every $j\in \Block{j'}$, that each element selected from $\Bucket{\LastN}{j}$ will not prevent the selection of any element from $S_i$ because 
$S_i\subseteq \Span{\Bucket{\FirstN}{\Gd{i}}} \subset \Span{\Bucket{\FirstN}{\Bd{j}}}$. 

Thus, when restricted to the elements of $\bigcup_{j\in \Block{i}} S_j$, the \GAPA\ can be viewed as if it was executing the \SA\ with input $J = \Block{i}\cup \Bd{i}$ and the elements revealed are those of
$S = \Bucket{\FirstN}{\Bd{i}}\cup \bigcup_{j\in \Block{i}} S_j$, which are revealed in an arbitrary order, except that the elements of $\Bucket{\FirstN}{\Bd{j}}$ are revealed first.
Thus, using Lemma~\ref{lem:greedy}, it is straight forward to see that
at least $\uncov{\Bucket{\FirstN}{\Bd{i}}\cup \bigcup_{j\in \Block{i}\setminus\{i'\}}{S_j}}{S_{i'}}$ are selected from $S_{i'}$, for every $i'\in \Block{i}$.
We shall show, that this term,  is at least $\uncov{\Bucket{\FirstN}{\Bd{j}}\cup 
\Bucket{\LastN}{\Block{i}\setminus\{i\}}}{S_{i'}}$, which in turn is at least
$\uncov{\Bucket{\FirstN}{\Bd{j}}\cup 
\Bucket{\LastN}{\Block{i}\setminus\{i\}}}{\Bucket{\LastN}{i'}}
- \indep{\Bucket{\FirstN}{\Gd{i}}}{\Bucket{\LastN}{i}}$.
In Section~\ref{sec:concentrations}, we show that with high probability,
by using only the elements of $\FirstN$, we can approximate $\uncov{\Bucket{\FirstN}{\Bd{j}}\cup 
\Bucket{\LastN}{\Block{i}\setminus\{i\}}}{\Bucket{\LastN}{i'}}$
and upper bound $\indep{\Bucket{\FirstN}{\Gd{i}}}{\Bucket{\LastN}{i}}$.
In Section~\ref{sec:structure}, we use the result of Section~\ref{sec:concentrations} to show that, if there is no element with a very high value, then either the \SA\ or the \GAPA\ will achieve the required \cor\, and we can choose the proper option using only the elements of $\FirstN$.

\begin{algorithm}
\caption{\GAPA}
\label{alg:Complicated}
{\bf Input: a \CT\  $(\BlockN,\GdN,\BdN)$}
\begin{enumerate}
	\item\label{step:initializeP} 
		$\PN\longleftarrow \emptyset$
	\item\label{step:reveal} 
		immediately after each element $\elementN\in \LastN$ is revealed do 
		\begin{enumerate}
  			\item\label{step:i} 
  				$\ell \longleftarrow \log{\val{\elementN}}$
  			\item\label{step:Special} 
  				if $\Block{\ell}\neq \emptyset$ do 		
			\begin{enumerate}
   				\item\label{step:FirstFilter} 
  				if $\elementN\in \Span{\Bucket{\FirstN}{\Gd{\ell}}}$,  do
				\begin{enumerate}
  					\item\label{step:SecondFilter} 
  					if $\elementN\not\in \Span{\PN\cup \Bucket{\FirstN}{\Bd{\ell}}}$, do 
  					$\PN\longleftarrow \PN\cup \{\elementN\}$
  				\end{enumerate}
  			\end{enumerate}
  		\end{enumerate}
\end{enumerate}
{\bf Output}: $\PN$
\end{algorithm}

\begin{lemma}\label{lem:disjoint}
Let  $i$ be such that
$\Block{i} \neq \emptyset$ and $\PN$ as it was in any stage in an arbitrary execution of the \GAPA.
If $\elementN\in \Bucket{\LastN}{\Block{i}}\cap\Span{\Bucket{\FirstN}{\Gd{i}}}$, then
$\elementN\not\in \Span{\PN\cup \Bucket{\FirstN}{\Bd{i}}}$ if and only if
$\elementN\not\in \Span{\left(\PN\cap \Bucket{\LastN}{\Block{i}}\right)\cup \Bucket{\FirstN}{\Bd{i}}}$.
\end{lemma}
\begin{proof}
Let $\elementN\in \Bucket{\LastN}{\Block{i}}\cap\Span{\Bucket{\FirstN}{\Gd{i}}}$.
We note that the "only if" condition trivially holds and hence we only prove the "if" condition.
Let $\PN^* = \PN\cap \Bucket{\LastN}{\Block{i}}$.
Assume that $\elementN\in \Span{\PN\cup \Bucket{\FirstN}{\Bd{i}}}$.
Let $C$ be a minimal subset of $\PN\setminus (\PN^*\cup \Span{\Bucket{\FirstN}{\Bd{i}}})$ such that $\elementN \in \Span{C \cup \Bucket{\FirstN}{\Bd{i}} \cup \PN^*}$.
We shall show that $C= \emptyset$ and hence $\elementN\in \Span{\left(\PN\cap \Bucket{\LastN}{\Block{i}}\right)\cup \Bucket{\FirstN}{\Bd{i}}}$.
Hence the result then follows.

Let $\elementN'$ be the latest element $C$ added to $\PN$ and let $j = \log{\val{\elementN'}}$.
According to construction, the elements of $C$ were selected by the \GAPA\ and hence
$\Block{j} \neq \emptyset$.
Also, by construction, $\Block{i}\neq \Block{j}$, since otherwise 
$\elementN' \in  \PN^* = \PN\cap \Bucket{\LastN}{\Block{i}}$.

Suppose that $\min{\Block{j}}>\max{\Block{i}}$.
By Items~\ref{item:BlockGood} and~\ref{item:CTLO} of Definition~\ref{def:tuple},
this implies that $\Block{j}\subseteq \Gd{j} \subseteq \Bd{i}$.
Since $\elementN'$ was selected by the \GAPA, by Step~\ref{step:FirstFilter}, this implies that
$\elementN' \in \Span{\Bucket{\FirstN}{\Gd{j}}} \subseteq \Span{\Bucket{\FirstN}{\Bd{i}}}$.
This contradicts the choice of $\elementN' \in  C \subseteq \PN\setminus (\PN^*\cup \Span{\Bucket{\FirstN}{\Bd{i}}})$.

Suppose on the other hand that $\max{\Block{j}}<\min{\Block{i}}$.
By Items~\ref{item:BlockGood} and~\ref{item:CTLO} of Definition~\ref{def:tuple},
this implies that $\Block{i}\cup \Bd{i} \subseteq \Gd{i} \subseteq \Bd{j}$ and
hence 
$\elementN\in \Bucket{\LastN}{\Block{i}}\cap\Span{\Bucket{\FirstN}{\Gd{i}}}\subseteq \Span{\Bucket{\FirstN}{\Bd{j}}}$ 
and, using Item~\ref{item:CTLO} of the definition of a \CT,
$\Bucket{\FirstN}{\Bd{i}} \cup \PN^* \subseteq
\Bucket{\FirstN}{\Bd{i}} \cup \Span{\Bucket{\FirstN}{\Gd{i}}} \subseteq
 \Span{\Bucket{\FirstN}{\Bd{j}}}$
 since, by Step~\ref{step:FirstFilter},
 every element in $\PN^*$ is in $\Span{\Bucket{\FirstN}{\Gd{i}}}$.

Since $\elementN'$ was the latest element in $C$ added to $\PN$,
by Item~\ref{item:MatroidExchange} of Proposition~\ref{prop:Matroid}, 
$\elementN' \in \Span{\{\elementN\}\cup\left(C \cup \Bucket{\FirstN}{\Bd{i}} \cup \PN^*\right) \setminus \{\elementN'\}}$.
Since $\elementN\in \Span{\Bucket{\FirstN}{\Bd{j}}}$ and
$\Bucket{\FirstN}{\Bd{i}} \cup \PN^* \subseteq  \Span{\Bucket{\FirstN}{\Bd{j}}}$, 
we see that
$\elementN' \in \Span{\left(C \cup \Bucket{\FirstN}{\Bd{j}}\right) \setminus \{\elementN'\}}$.
Therefore,
$\elementN'$ did not satisfy the condition in Step~\ref{step:SecondFilter}. 
This contradicts the fact that $\elementN'$
was  added to $\PN$.
\end{proof}

\begin{theorem}\label{thm:tuple}
Given a \CT\ $(\BlockN,\GdN,\BdN)$ as input,
Algorithm~\ref{alg:Complicated} returns an independent set of elements $\PN\subseteq \LastN$ such that 
$$\OPT{\PN}
\geq 
\sum_{j\in \BLOCK}2^j\cdot\left(\uncov{\Bucket{\FirstN}{\Bd{j}}\cup \Bucket{\LastN}{\Block{j}\setminus\{j\}}}{\Bucket{\LastN}{j}} 
- 
\indep{\Bucket{\FirstN}{\Gd{j}}}{\Bucket{\LastN}{j}}\right).$$
\end{theorem}
\begin{proof}
Step~\ref{step:SecondFilter} implies that $\PN$ is always an independent set.
Let $j\in \BLOCK$ and,
for every $i\in \Block{j}$, let 
$S_i = \Bucket{\LastN}{i}\cap \Span{\Bucket{\FirstN}{\Gd{i}}}$.

We note that, by definition, for every $i\in \Block{j}$, every element in $S_i$ satisfies the condition in Step~\ref{step:FirstFilter}.
Consequently, by Lemma~\ref{lem:disjoint}, the \GAPA\ processes the elements in $\bigcup_{i\in \Block{j}} S_i$, as if it was the \SA\ in the following setting:
the input is a set $J = \Block{j}\cup \Bd{j}$ and the elements revealed are those of
$S = \Bucket{\FirstN}{\Bd{j}}\cup \bigcup_{i\in \Block{j}} S_i$, which are revealed in an arbitrary order, except that  the elements of $\Bucket{\FirstN}{\Bd{j}}$ are revealed first.
Thus, by Lemma~\ref{lem:greedy},
$\rank{\Bucket{\PN}{j}}$ is at least $\uncov{\Bucket{\FirstN}{\Bd{j}}\cup \bigcup_{i\in \Block{j}\setminus\{j\}}{S_i}}{S_j}$.

Let $\Bucket{\LastN}{j}$,
$\RN_1 = \Bucket{\FirstN}{\Bd{j}}\cup \bigcup_{i\in \Block{j}\setminus\{j\}}S_i$ and 
$\RN_2 = \Bucket{\FirstN}{\Bd{j}}\cup \Bucket{\LastN}{\Block{j}\setminus\{j\}}$.
Then, 
$\rank{\Bucket{\PN}{j}} \geq \uncov{\RN_1}{\SN_j}.$
Since $\RN_1 \subseteq \RN_2$, by Observation~\ref{obs:uncov}, we have that 
 $\uncov{\RN_1}{\SN_j}\geq \uncov{\RN_2}{\SN_j}$.
By definition, $\uncov{\RN_2}{\SN_j}=\uncov{\RN_2}{\Bucket{\LastN}{j}} - \left(\rank{\RN_2 \cup \Bucket{\LastN}{j}} - \rank{\RN_2 \cup \SN_j}\right)$. 
Finally, 
$\SN_j = \Bucket{\LastN}{j}\cap \Span{\Bucket{\FirstN}{\Gd{i}}}$, 
we see that 
$\rank{\RN_2 \cup \Bucket{\LastN}{j}} - \rank{\RN_2 \cup \SN_j}$ does not exceed
$\rank{\Bucket{\LastN}{j}\setminus \Span{\Bucket{\FirstN}{\Gd{i}}}}=
\indep{\Bucket{\FirstN}{\Gd{i}}}{\Bucket{\LastN}{j}}$.
Therefore,  $\uncov{\RN_2}{\SN_j}\geq\uncov{\RN_2}{\Bucket{\LastN}{j}} - \indep{\Bucket{\FirstN}{\Gd{i}}}{\Bucket{\LastN}{j}}$.
Thus, by Observation~\ref{obs:OPT}, the result follows.
\end{proof}

\section{Prediction}\label{sec:concentrations}
In this section, we prove that
for a specific subset of the integers, which we denote by $\Super$ and later, with constant probability,   for every
 $\KSet,\KSetStar\subseteq \Super$, where $\max{\KSetStar} < \min{\KSet}$ and
 $\min\{\rank{\BucketU{i}}\mid i\in \KSet\} \geq \rank{\BucketU{\min{\KSet}}}^{\unionBoundGapPower}$, and
 for every $j\in \KSet$ we have that
 (i) $\uncov{\Bucket{\FirstN}{\KSetStar}\cup \Bucket{\LastN}{\KSet\setminus\{j\}}}{\Bucket{\LastN}{j}}$
  is approximately
 $\uncov{\Bucket{\FirstN}{\KSetStar\cup\KSet\setminus\{j\}}}{\Bucket{\FirstN}{j}}$; 
 and (ii)
 $\indep{\Bucket{\FirstN}{\KSet}}{\Bucket{\LastN}{j}}$ 
 is bounded above by approximately 
$\uncov{\Bucket{\FirstN}{\KSet\setminus\{j\}}}{\Bucket{\FirstN}{j}}$.
This result enables us at \StageTwo\ of the \MA\ to chose whether to select elements using the \SA\ or the \GAPA, and to compute the input to the chosen algorithm.

In Subsection~\ref{subsec:Talagrand}, we use the Talagrand inequality for the unquantified version of (i),
in Subsection~\ref{subsec:Closure}, we use Martingales and Azuma's inequality fur the unquantified version of (ii) and
in Subsection~\ref{subsec:UnionBound}, we define the set $\Super$ and use the Union Bound together with the results in the previous sections to prove the main result of this section.

\subsection{Upper Bounding $\indepN$}\label{subsec:Closure}
 
\begin{theorem}\label{thm:closure}
Let $\KKSet \subset \ZZ$, be finite and non-empty and $k\in \KKSet$  
   then,
  $$\Prob{\indep{\Bucket{\FirstN}{\KKSet}}{\Bucket{\LastN}{k}}\leq \uncov{\Bucket{\FirstN}{\KKSet\setminus\{k\}}}{\Bucket{\FirstN}{k}} +   \AzzumaDevMultCnst\cdot\rank{\BucketU{k}}^{\rZZDiff}}>
  1-e^{-\AzzumaMultCnst \rank{\BucketU{k}}^{\AzzumaPowerCnst}}.$$  
\end{theorem} 
 
\begin{proof}
We fix $S = \Bucket{\FirstN}{\KKSet\setminus\{k\}}$ and let
$m = \rank{\BucketU{k}}$.
We initially let both $\HSetFirst$ and $\HSetLast$ be empty sets.
Then, we repeat the following $\numberOfIter m$ times:
if there exists an element in $\BucketU{k} \setminus (\HSetFirst \cup \HSetLast)$ that is not in $\Span{S\cup \HSetFirst}$, then we pick such an element arbitrarily, if it is in $\FirstN$, then we add it to $\HSetFirst$ and otherwise we add it to $\HSetLast$.

We observe that every time an element is added to $\HSetFirst$ it is independent of $\Span{S\cup \HSetFirst}$ and hence it increases by one the quantity $\uncov{S}{\HSetFirst} = \rank{S\cup \HSetFirst} - \rank{S}$.
Thus, if after $\numberOfIter m$ repetitions there are no elements in $\BucketU{k}\setminus \Span{S\cup \HSetFirst}$, then
the preceding quantity cannot be increased further by adding elements from 
$\BucketU{k} \setminus (\HSetFirst \cup \HSetLast)$
to
$\HSetFirst$
and therefore
$\uncov{S}{\HSetFirst} = |\HSetFirst|$.
Since, in this case every element in $\Bucket{\FirstN}{k}\setminus \HSetFirst$ is in $\Span{S\cup \HSetFirst}$ and $\HSetFirst\subseteq \Bucket{\FirstN}{k}$, we see that
$\Span{S \cup \Bucket{\FirstN}{k}} = \Span{S\cup \HSetFirst}$.
Therefore,
$\rank{S \cup \Bucket{\FirstN}{k}} = \rank{S\cup \HSetFirst}.$
Hence, by the definition of \uncovN, $\uncov{S}{\Bucket{\FirstN}{k}} = \uncov{S}{\HSetFirst} = |\HSetFirst|$.

We also observe that every time an element is added to $\HSetLast$ it  may increase by one the quantity $\indep{S\cup \HSetFirst}{\HSetLast} = \rank{\HSetLast\setminus \Span{S\cup \HSetFirst}}$.
We note that if after $\numberOfIter m$ repetitions there are no elements in $\BucketU{k}\setminus \Span{S\cup \HSetFirst}$, then
the preceding quantity cannot be increased further
by adding elements from 
$\BucketU{k} \setminus (\HSetFirst \cup \HSetLast)$
to
$\HSetLast$
 and therefore
$\indep{S\cup \HSetFirst}{\HSetLast} \leq |\HSetLast|$.
Since, in this case every element in 
$\Bucket{\LastN}{k}\setminus \HSetLast$ is in $\Span{S\cup \HSetFirst}$ and $\HSetLast \subseteq \Bucket{\LastN}{k}$, we see that  
$\indep{S\cup \HSetFirst}{\HSetLast} = \indep{S\cup \HSetFirst}{\Bucket{\LastN}{k}}$.
We note that $S\cup \Bucket{\FirstN}{k} = \Bucket{\FirstN}{\KKSet}$ and we already proved $\Span{S \cup \Bucket{\FirstN}{k}} = \Span{S\cup \HSetFirst}$.
Thus,
$\indep{\Bucket{\FirstN}{\KKSet}}{\Bucket{\LastN}{k}} = \indep{S\cup \HSetFirst}{\HSetLast} \leq |\HSetLast|$.

Next we show that,
$||\HSetFirst| - |\HSetLast|| \leq \AzzumaDevMultCnst\cdot m^{\rZZDiff}$,
 with probability at least $1-e^{-\AzzumaMultCnstTT m^{\AzzumaPowerCnst}}$, and afterwards we show that,
  with probability at least $1- e^{-\frac{m}{2}}$, after $\numberOfIter m$ repetitions, there are no elements in $\BucketU{k}\setminus \Span{S\cup \HSetFirst}$.
By the union bound, this implies the theorem.

We define the variables $\RVZN_i$ so that
 $\RVZN_0 = 0$ and
(i) $\RVZN_i = \RVZN_{i-1}-1$ if in the i'th repetition an element was added to $\HSetFirst$;
(ii) $\RVZN_i = \RVZN_{i-1}+1$ if in the i'th repetition an element was added to $\HSetLast$; and 
(iii) $\RVZN_i = \RVZN_{i-1}$ if nothing happened in the i'th repetition.

We note that, for every $i>0$,  either $\RVZN_i = \RVZN_{i-1}$ or $\RVZN_i$ is distributed uniformly over $\{\RVZN_{i-1}-1,\RVZN_{i-1}+1\}$ and hence $E(\RVZN_i\mid \RVZN_{i-1}) = \RVZN_{i-1}$, where $E()$ denotes the expected value.
Consequently, we have a martingale.
Thus, by Azuma's inequality, $\RVZN_{\numberOfIter m} > \AzzumaDevMultCnst\cdot m^{\rZZDiff}$ 
with probability less than
$e^{-\AzzumaMultCnstTT \cdot m^{\AzzumaPowerCnst}}$.
Since, $\RVZN_{\numberOfIter m} = |\HSetLast| - |\HSetFirst|$ we have proved the first inequality. We now proceed to the second.

We define the variables $\RVXN_i$ so that $\RVXN_i = 1$ if in the $i$th repetition the element processed was in $\FirstN$ and otherwise $\RVXN_i = 0$. 
By definition, for every $i\in [\numberOfIter m]$, if 
$\RVZN_i = \RVZN_{i-1} - 1$, then $\RVXN_i=1$.
If $|\HSetFirst|=m$ after $\numberOfIter m$ repetitions, then
$\rank{\SN\cup \HSetFirst} = \rank{\SN} + \rank{\BucketU{k}}$, which can only happen if 
$\BucketU{k} \subseteq \Span{\SN\cup \HSetFirst}$. 
This implies that
$\BucketU{k} \setminus \Span{\SN\cup \HSetFirst}$ is empty.
So $\BucketU{k} \setminus \Span{\SN\cup \HSetFirst}$ is not empty after 
$\numberOfIter m$ only if $\sum_{i=1}^{\numberOfIter m}\RVXN_i<m$.
By Observation~\ref{obs:binomial}, for every $i\in [\numberOfIter m]$,
$\RVXN_i$ is independently distributed uniformly over $\{0,1\}$.
By the Chernoff inequality,
with probability at least $1 - e^{-\frac{m}{2}}$,
 $\sum_{i=1}^{\numberOfIter m} \RVXN_i \geq m$.
\end{proof}

\subsection{Talagrand based concentrations}\label{subsec:Talagrand}
This subsection is very similar to one that appears in~\cite{improved}, we include it for the sake of completeness.
 The following definition is an adaptation of the Lipschitz condition to our setting.
 
\begin{definition}{\textbf{[Lipschitz]}}~\label{def:Lipschitz}
Let  $f:\UU \longrightarrow \NN$.
If $|f(\SN_1) - f(\SN_2)|\leq 1$ for every $\SN_1,\SN_2\subseteq \UU$ such that  $|(\SN_1\setminus \SN_2) \cup (\SN_2\setminus \SN_1)| =1$,
 then $f$ is \textbf{Lipschitz}. 
\end{definition}

\begin{definition}{\textbf{[Definition 3, Section 7.7 of \cite{alon}]}}
Let $f:\NN\longrightarrow \NN$. $h$ is {\it $f$-certifiable} if whenever $h(x) \geq s$ there exists $I \subseteq \{1,\dots,n\}$ with $|I|\leq f(s)$ so that all $y\in \Omega$ that agree with $x$ on the coordinates $I$ have $h(y)\geq s$.
\end{definition}
 
\begin{obs}~\label{obs:certifiable}
For every finite $\USet\subset \ZZ$,
the $\rankN$ function over subsets of $\BucketU{\USet}$ is Lipschitz and $f$-certifiable with $f(s) = \rank{\BucketU{\USet}}$, for all $s$.
\end{obs}
\begin{proof}
The $\rankN$ function is Lipschitz, by the definition of the $\rankN$ function (Definition~\ref{def:RankSpan}).
By Item~\ref{item:MatroidClContained} of Proposition~\ref{prop:Matroid},
for every $ \SN \subseteq \RN \subseteq \BucketU{\USet}$, we have that
$\rank{\SN} \leq \rank{\RN} \leq \rank{\BucketU{\USet}}$.
Thus, the $\rankN$ function over subsets of $\BucketU{\USet}$ is $f$-certifiable with $f(s) = \rank{\BucketU{\USet}}$.
\end{proof}
 
The succeeding theorem is a direct  result of Theorem 7.7.1 from \cite{alon}.
\begin{theorem}\label{thm:Talagrand}
If $h$ is Lipschitz and $f$ certifiable, then for $x$ selected uniformly from $\Omega$ and all $b,t$,
$Pr[h(x)\leq b - t\sqrt{f(b)}]\cdot Pr[h(x) \geq b] \leq e^{-t^2/4}.$
\end{theorem}

\begin{lemma}~\label{lem:Talagrand}
Let $t\geq 2$, $j\in \ZZ$,
$\KSet,\KSetStar\subseteq \ZZ$, where $\min{\KSetStar} < \max{\KSet}$ and
 $k\in \KSet$ then, 
 $
 \Prob{
   \Big|\rank{\Bucket{\FirstN}{j}} 
   -
  \rank{\Bucket{\LastN}{j}}\Big| \geq 
  2 t\sqrt{\rank{\BucketU{j}}}}
  \leq 
  e^{1.4-\frac{t^2}{4}}$ 
  and 
  
 $$\Prob{
  \Big|\uncov{\Bucket{\FirstN}{\KSetStar}\cup \Bucket{\LastN}{\KSet\setminus\{k\}}}{\Bucket{\LastN}{k}} 
  -
 \uncov{\Bucket{\FirstN}{\KSetStar\cup\KSet\setminus\{k\}}}{\Bucket{\FirstN}{k}}\Big| \geq 
 \TalagrandMultCnst t\sqrt{\rank{\BucketU{}}}}
 \leq 
 e^{\TalagrandExpAdditiveCnst-\frac{t^2}{4}}.$$
\end{lemma}
\begin{proof}
Let $S\in\{\Bucket{\FirstN}{\KSetStar\cup\KSet},
\Bucket{\FirstN}{\KSetStar\cup\KSet\setminus\{k\}},
\Bucket{\FirstN}{\KSetStar}\cup\Bucket{\LastN}{\KSet},
\Bucket{\FirstN}{\KSetStar}\cup \Bucket{\LastN}{\KSet\setminus\{k\}}\}$.
By Observation~\ref{obs:certifiable}, the $\rankN$ function is Lipschitz and $\rankN$-certifiable.

Clearly,
since $\FirstN$ and $\LastN$ are both distributed uniformly, with probability at least $\frac{1}{2}$, we have that 
 $\rank{S}\geq \med{\rank{S}}$.
Hence, taking 
$b  = \med{\rank{S}} + t\sqrt{\rank{\BucketU{\KSetStar\cup\KSet}}}$, 
 by Theorem~\ref{thm:Talagrand}, we get that
$\rank{S} - \med{\rank{S}} \geq t\sqrt{\rank{\BucketU{\KSetStar\cup\KSet}}}$,
with probability at most $2e^{\frac{-t^2}{4}}$.
In a similar manner, by taking $b  = \med{\rank{S}}$,
we get that
$\med{\rank{S}}
-
\rank{S} \geq t\sqrt{\rank{\BucketU{\KSetStar\cup\KSet}}}$,
with probability at most $2e^{\frac{-t^2}{4}}$.
Thus, by the union bound,
$|\rank{S}-\med{\rank{S}}| \geq t\sqrt{\rank{\BucketU{\KSetStar\cup\KSet}}}$,
with probability at most $4e^{-\frac{t^2}{4}}$. 

We note that, since $\FirstN$ and $\LastN$ are identically distributed and $K\cap K^* = \emptyset$, we have that
$\med{\rank{\Bucket{\FirstN}{\KSetStar\cup\KSet}}} = \med{\rank{\Bucket{\FirstN}{\KSetStar}\cup\Bucket{\LastN}{\KSet}}}$
and 
$\med{\rank{\Bucket{\FirstN}{\KSetStar\cup\KSet\setminus\{k\}}}} = \med{\rank{\Bucket{\FirstN}{\KSetStar}\cup\Bucket{\LastN}{\KSet\setminus\{k\}}}}$
.
Consequently, the second part of the result follows, by the union bound and the definition of \uncovN\ (Definition~\ref{def:uncov}).
The first part follows in a similar manner the preceding analysis.
\end{proof}

\subsection{Union bound}\label{subsec:UnionBound}
\begin{definition}\textbf{[\Super]}\label{def:SuperBuckets}
 We define $\Super = \left\{i \Big|\rank{\BucketU{i}} > 
 \left(2^{-i\PSuperCnst}\cdot \LOPT{\UU}\right)^{\SuperPwr}\right\}.$
\end{definition}

When the following theorem is used later, the notations
$\KSetStar$ and $\KSet$ are replaced once with 
$\Bd{i}$ and $\Block{i}$, respectively,
 another time with the empty set and the input to the \SA, respectively.

\begin{theorem}\label{thm:Concentrations}
If $\rank{\UU} > 2^{\maxWeightPower}$
then,
with probability at least $\TradeoffProb$, 
the following event holds:
for every $i\in \Super$,
$\KSet,\KSetStar\subseteq \{j\in \Super \mid j \leq \log\LOPT{\UU} - \SuperSplitMultCnst\cdot \log\log{\rank{\UU}}\}$,
where  $\min{\KSetStar} > \max{\KSet}$ or $K' =\emptyset$, 
$\min\{\rank{\BucketU{j}}\mid j\in \KSet\} \geq \left(\concentrationRangeMultCnst\cdot\rank{\BucketU{\min{\KSet}}}\right)^{\manageablePower}$, and
every $k\in \KSet$,
the following hold:
\begin{enumerate}
\item\label{item:rankFirst}
$4\cdot \rank{\FirstN} > \rank{\UU},$
\item \label{item:SingleBucket}
$\Big|\rank{\Bucket{\FirstN}{i}} - \rank{\Bucket{\LastN}{i}}\Big| < \TalagrandMultCnst\cdot\rank{\BucketU{i}}^{\singleUnionPowerPHalf},
$
\item \label{item:ConcentrationUncov}
$
 \Big|\uncov{\Bucket{\FirstN}{\KSetStar}\cup \Bucket{\LastN}{\KSet\setminus\{k\}}}{\Bucket{\LastN}{k}} 
-
\uncov{\Bucket{\FirstN}{\KSetStar\cup\KSet\setminus\{k\}}}{\Bucket{\FirstN}{k}}\Big| 
<
\TalagrandMultCnstTT\cdot\rank{\BucketU{\KSetStar\cup\KSet}}^{\simpleGapRankExp},
$
\item\label{item:ConcentrationIndep} 
$\indep{\Bucket{\FirstN}{\KSet}}{\Bucket{\LastN}{k}}
 \leq
    \uncov{\Bucket{\FirstN}{\KSet\setminus\{k\}}}{\Bucket{\FirstN}{k}} + \AzzumaDevMultCnstTT\cdot \rank{\BucketU{k}}^{\rZZDiff}.$
\end{enumerate}
\end{theorem}
\begin{proof}
Let $C$ be a maximal independent set in $\UU$.
By Observation~\ref{obs:binomial} and the Chernoff bound,
$\Prob{|\FirstN \cap C| \leq \frac{1}{4}\cdot \rank{\UU}} \leq e^{-2^{-3}\cdot \rank{\UU}} < \frac{1}{8}$, where the last inequality follows from $\rank{\UU} > 2^{\maxWeightPower}$.
By the definition of $\rankN$, $|\FirstN \cap C| > \frac{1}{4}\cdot \rank{\UU}$ implies Item~\ref{item:rankFirst}.

Let 
$c = \log\LOPT{\UU} - \SuperSplitMultCnst\cdot \log\log{\rank{\UU}}$, 
$\KSet,\KSetStar\subseteq \Super$, where 
$\min{\KSetStar} > \max{\KSet}$  and
$\min\{\rank{\BucketU{i}}\mid i\in \KSet\} \geq \rank{\BucketU{\min{\KSet}}}^{\unionBoundGapPower}$,   
$k\in \KSet$ 
and 
$t = 2\cdot\rank{\BucketU{\KSetStar\cup\KSet}}^{\simpleUnionPower} \geq 2$.

Consequently, by the union bound, Theorem~\ref{thm:closure} and
Lemma~\ref{lem:Talagrand}, at least one of Items~\ref{item:ConcentrationUncov} and~\ref{item:ConcentrationIndep} does not hold for $\KSetStar,\KSet$ and $k$, with probability at most
$e^{\TalagrandExpAdditiveCnst-\frac{t^2}{4}} + e^{-\AzzumaMultCnst \rank{\BucketU{k}}^{\AzzumaPowerCnst}}$, which 
 does not exceed
$e^{\TalagrandExpAdditiveCnst-\rank{\BucketU{\KSetStar\cup\KSet}}^{\simpleUnionPowerTT}} + e^{-\AzzumaMultCnst \rank{\BucketU{k}}^{\AzzumaPowerCnst}}$ 
which, in turn, is less than
$e^{\TalagrandExpAdditiveCnstPO-\left(\concentrationRangeMultCnst\cdot\rank{\BucketU{\min{\KSet}}}\right)^{\unionBoundGapPowerThird}}$, 
because $k\in \KSet$ and 
$\concentrationRangeMultCnst\cdot\min\{\rank{\BucketU{i}}\mid i\in \KSet\} \geq \rank{\BucketU{\min{\KSet}}}^{\manageablePower}$. 
By the definition of $\Super$, we see that 
$e^{\TalagrandExpAdditiveCnstPO-\left(\concentrationRangeMultCnst\cdot\rank{\BucketU{\min{\KSet}}}\right)^{\unionBoundGapPowerThird}} \leq e^{\TalagrandExpAdditiveCnstPO-\left(2^{ \log{\LOPT{\UU}}-\min\KSet\PSuperCnstMSeven}\right)^{\unionBoundGapPowerSixth}}$.

Let $z \in \Super$.
Since $z\in \KSet\subseteq \Super$, and according to the definition of $\KSet$ and $\KSetStar$, for every possible value of $z$ there are at most 
$|\KSet| \leq c-z$ possible choices of $k$ and $(c-z)2^{c-z}$ choices of $\KSet$ and $\KSetStar$.
Consequently, by the union bound, the probability that
at least one of Items~\ref{item:ConcentrationUncov} and~\ref{item:ConcentrationIndep} does not hold for some $\KSetStar,\KSet$ and $k$, is at most
$\sum_{z\in\Super}(c-z)^22^{c-z}e^{\TalagrandExpAdditiveCnstPO-\left(2^{\log{\LOPT{\UU}}-z\PSuperCnstMSeven}\right)^{\unionBoundGapPowerSixth}}$.
Taking $y = c - z$, the previous value is bounded above by
$\sum_{y\in\NN}y^22^ye^{\TalagrandExpAdditiveCnstPO-\left(2^{\log{\LOPT{\UU}}-c+y\PSuperCnstMSeven}\right)^{\unionBoundGapPowerSixth}}<\frac{1}{8}$, 
since  $c =  \log\LOPT{\UU} - \SuperSplitMultCnst\cdot \log\log{\rank{\UU}}$ and $\rank{\UU} > 2^{\maxWeightPower}$.

Let $k'\in \Super$, and 
$t' = 2\cdot\rank{\BucketU{k'}}^{\singleUnionPower} \geq 2$.
By Lemma~\ref{lem:Talagrand}, Item~\ref{item:SingleBucket} does not hold,
with probability at most
$e^{1.4-\rank{\BucketU{k'}}^{\singleUnionPowerTT}}.$
Since $\rank{\BucketU{k'}} \geq 
\left(2^{-k'\PSuperCnst}\cdot \LOPT{\UU}\right)^{\SuperPwr}
$, 
 by the definition of $\Super$, we see that 
$e^{1.4-\rank{\BucketU{k'}}^{\singleUnionPowerTT}} \leq e^{1.4-\left(2^{ \log{\LOPT{\UU}}-k'\PSuperCnst}\right)^{\singleUnionPowerTTSuperPwr}}$.
Therefore, 
by the union bound, the probability that
 Item~\ref{item:SingleBucket} does not hold for any $k'\in \Super$, is at most
 $\sum_{k'\in\Super} e^{1.4-\left(2^{ \log{\LOPT{\UU}}-k'\PSuperCnst}\right)^{\singleUnionPowerTTSuperPwr}}$.
Taking $y' = \max{\Super} - k'$, the previous value is bounded above by
$\sum_{y'\in\NN} e^{1.4-\left(2^{ \log{\LOPT{\UU}}-\max{\Super}+y'\PSuperCnst}\right)^{\singleUnionPowerTTSuperPwr}}<\frac{1}{2}$, 
where the last inequality follows from Assumption~\ref{ass:MaxVal}.
Consequently, by the union bound the result follows.
\end{proof}

\section{Structural Theorem}\label{sec:structure}
In this section we assume that all the elements of $\FirstN$ have been revealed and hence $\FirstN$ is treated as fixed.

\begin{definition}\textbf{[$\GTSetN_H$]}\label{def:GTSet}
For every $\USet\subseteq H \subset \ZZ$ be let
$\GTSetH{\USet} = \{i\in H\mid i > \max{\USet}\}.$
We omit the subscript when clear from context.
\end{definition}

\begin{definition}\textbf{[\manageable\ set]}\label{def:manageable}
A set of integers $K$ is \textbf{\manageable} if, for every $j\in K$,
we have that $\rank{\Bucket{\FirstN}{j}} \geq \left(\manageableMultCnst\cdot\sum_{i\in K} \rank{\Bucket{\FirstN}{i}}\right)^{\manageablePower}$.
\end{definition}

\begin{definition}\textbf{[\Cf]}\label{def:Cf}
Let $L\subset \ZZ$, let $\mFam$ be a family of subsets of $L$, and let
$H = \bigcup_{H'\in \mFam}H'$,
then $\mFam$ is a \cf\ for $L$ if the following hold:
\begin{enumerate}
\item\label{item:CfLOPT} $\LOPT{\Bucket{\FirstN}{H}} \geq \frac{1}{\SSDropCnst}\cdot \LOPT{\Bucket{\FirstN}{L}}$,
\item\label{item:CfOrdered} for every pair $H_1$ and $H_2$ of distinct sets in $\mFam$, either $\max{H_1} <  \min{H_2}$ or 
$\max{H_2}< \min{H_1}$,
\item\label{item:CfUncov} for every $H' \in \mFam$ and $j\in H'$, 
$\uncov{\Bucket{\FirstN}{\GTSetH{H'}}}{\Bucket{\FirstN}{j}} 
 \geq \frac{\IncMultCnstHalfMO}{\IncMultCnstHalf}\cdot\rank{\Bucket{\FirstN}{j}},$ 
\item\label{item:CfManageable} every set in $\mFam$ is \manageable\ and
\item\label{item:CfMinIndexBucket}
for every $i\in \HN$, $2\cdot\rank{\Bucket{\FirstN}{i}} \geq \rank{\Bucket{\FirstN}{\GTSetH{\{i\}}\cup \{i\}}}$.
\end{enumerate}
\end{definition}

\begin{lemma}\label{lem:Cf}
Let $L\subset \ZZ$.
If $\rank{\FirstN} > 2^{\maxWeightPowerMT}$, then
there exists a \cf\ $\mFam$ for $L$ of cardinality at most $\mFamCardUBCnst\cdot\log\log{\rank{\FirstN}}$.
\end{lemma}

\begin{proof}
Define,
$\NLMap{}:L \longrightarrow \NN$
as follows:
for every $i\in L$,
$\NLMap{i} = \rank{\Bucket{\FirstN}{i}}$. 
Let 
$m = \sum_{j\in L}\NLMap{j}\cdot 2^{j} = \LOPT{\Bucket{\FirstN}{L}}$.
Let $s_1$ be maximum so that $\NLMap{s_1} > 0$, and inductively define, $s_{i+1}$ to be the maximum integer such that   
$\NLMap{s_{i+1}}\geq 2\cdot\NLMap{s_{i}}$.
Let $k$ be the maximum integer such that $s_{k}$ is defined.
It follows that $s_1 > s_2 > \dots > s_k$.
For every $i\in [k]$, the sum of $\NLMap{j}\cdot 2^j$ over all $j\in L$, where $s_{i+1}<j<s_i$ when $i < k$, is at most
$2\cdot  \NLMap{s_{i}}\cdot\growthCnst^{s_{i}}$.
Thus, $m \leq 3\cdot\sum_{i=1}^{k}\NLMap{s_{i}}\cdot \growthCnst^{s_{i}}$ and so
$\sum_{i=1}^{k}\NLMap{s_{i}}\cdot \growthCnst^{s_{i}} \geq \frac{m}{\NLDropCnst}$.

Let $R_j = \{s_i \mid i = j {\mod{5}} \}$.
By the pigeon hole principle, there exists $q \in [\TeleCnst]$ such that
the sum of $\sum_{i\in R_q} 2^i \cdot\NLMap{s_i}\geq \frac{m}{\SSDropCnst}$.
Let $\ell_1 = \max{R_q}$ and $r_1$ be the minimum member of $R_q$ such that $\NLMap{\ell_1} \geq \NLMap{r_1}^{\manageablePower}$ and set $H_1 = [\ell_1,r_1]\cap R_q$.
Now, inductively, for every $i>1$, let $\ell_{i}$ be the maximum member of $R_q$ that is smaller than $r_{i-1}$, and $r_{i}$ be the minimum member of $R_q$ such that $\NLMap{\ell_{i}} \geq \NLMap{r_i}^{\manageablePower}$ and $H_i = [\ell_i,r_i]\cap R_q$.
Let $g$ be the maximum integer for which $H_{g}$ is defined, $\mFam = \{H_i\}_{i\in [g]}$ and
$H = \bigcup_{i=1}^{g} H_i$.
We note that, by construction, $r_g = \min{R_q}$ and $H = R_q$.

We next bound above $g$.
If $g \leq 4$, then $g \leq \mFamCardUBCnst\cdot\log\log{\rank{\FirstN}}$, since $\rank{\FirstN} \geq 2^{\maxWeightPowerMT}$.
Suppose that $g > 4$.
By construction, $\NLMap{\ell_i} < \NLMap{\ell_{i+1}}^{\manageablePower}$, for every $i\in[g-1]$,  and hence
$\rank{\FirstN} \geq\NLMap{\ell_{g}} >
\NLMap{\ell_2}^{ (\manageablePowerInv)^{g -2}}.$
Since, by construction, 
$\NLMap{\ell_2} \geq \IncMultCnst\cdot\NLMap{\ell_1} \geq \IncMultCnst$,
 the preceding inequality and the fact that 
$\rank{\FirstN} > 2^{\maxWeightPowerMT}$ 
  imply that  
$g \leq \PartitionCardCnst\cdot \log\log{\rank{\FirstN}}$. 

By construction,
\textbf{Items~\ref{item:CfOrdered} and~\ref{item:CfMinIndexBucket} of Definition~\ref{def:Cf}} holds.
Also
 $\sum_{p\in H} 2^p \cdot \NLMap{s_p} = \sum_{p\in R_q} 2^p \cdot \NLMap{s_p}\geq \frac{m}{\SSDropCnst}$.
so \textbf{Item~\ref{item:CfLOPT} of Definition~\ref{def:Cf}} holds. 
Let $H' \in \mFam$ and $j\in H'$.
By construction,
$\rank{\Bucket{\FirstN}{j}} = \NLMap{j} \geq \NLMap{\min{H'}}^{\manageablePower} \geq
\left(\frac{1}{2}\cdot\sum_{p\in  H'}\NLMap{p}\right)^{\manageablePower} = \left(\frac{1}{2}\cdot\sum_{p\in  H'}\Bucket{\FirstN}{p}\right)^{\manageablePower}$.
Thus, \textbf{Item~\ref{item:CfManageable} of Definition~\ref{def:Cf}} holds.

By the definition of $\uncovN$ (Definition~\ref{def:uncov}),
$\uncov{\rank{\Bucket{\FirstN}{\GTSetH{H'}}}}{\rank{\Bucket{\FirstN}{j}}} = \rank{\Bucket{\FirstN}{\GTSetH{H'}\cup \{j\}}} - \rank{\Bucket{\FirstN}{\GTSetH{H'}}}\geq \rank{\Bucket{\FirstN}{j}} - \rank{\Bucket{\FirstN}{\GTSetH{H'}}}$.
This is bounded below by $\frac{\IncMultCnstHalfMO}{\IncMultCnstHalf}\cdot\rank{\Bucket{\FirstN}{j}}$ because, by construction, for every $j\in H'$
we have 
$\rank{\Bucket{\FirstN}{j}} 
\geq 
\IncMultCnstHalf\cdot \sum_{i\in \GTSetH{H'}}\rank{\Bucket{\FirstN}{i}} 
\geq 
\IncMultCnstHalf\cdot\rank{\Bucket{\FirstN}{\GTSetH{H'}}} 
$.
Consequently, \textbf{Item~\ref{item:CfUncov} of Definition~\ref{def:Cf}} holds.
\end{proof}

\begin{definition}\textbf{[\usefulN]}\label{def:useful}
Let $\USet^*\subseteq \USet\subseteq H\subset \ZZ$.  
If the following hold:
\begin{enumerate}
\item\label{item:usefulOPT}
$\LOPT{\Bucket{\FirstN}{\USet^*}} > \usefulLBCnst\cdot\LOPT{\Bucket{\FirstN}{\USet}}$ and
\item\label{item:usefulUncov}
$\sum_{j\in \USet^*}2^j\cdot\left(
\uncov{\Bucket{\FirstN}{\GTSetH{\USet^*}\cup \USet^*\setminus\{j\}}}{\Bucket{\FirstN}{j}} 
  			 	-
\uncov{\Bucket{\FirstN}{\GTSetH{\USet}\cup \USet\setminus\{j\}}}{\Bucket{\FirstN}{j}}\right) 
  			 	\geq
\frac{\LOPT{\Bucket{\FirstN}{\USet}}}{\UsefulDenumCnst\cdot\log\log{\rank{\FirstN}}}.		
$
\end{enumerate}
then $\USet^*$ is \textbf{\usefulN} for $\USet$ in $H$ 
and $\USet$ is \textbf{\usefulN} in $H$.
\end{definition}

\begin{definition}\textbf{[\splittable]}\label{def:splittable}
$\USet\subset \ZZ$ is \textbf{\splittable} if it has a bipartition $\{\USet_1,\USet_2\}$ such that
\begin{enumerate}
\item $\min{\USet_1}>\max{\USet_2}$ and
\item\label{item:splittableVal} $\LOPT{\Bucket{\FirstN}{\USet_i}} > \splitLBCnst\cdot\LOPT{\Bucket{\FirstN}{\USet}}$, for every $i= 1,2$.
\end{enumerate}
\end{definition}

\begin{definition}\textbf{[\negligible]}\label{def:negligible}
A subset $\USet$ of a set $H\subset \ZZ$ is \textbf{\negligible} for $H$ if $\LOPT{\Bucket{\FirstN}{\USet}} <  \frac{\LOPT{\Bucket{\FirstN}{H}}}{\preLOPTLBDiv\cdot\log{\rank{\FirstN}}}$.
When $H$ is clear from context, we just say $K$ is \negligible. 
\end{definition}

\begin{definition}\textbf{[\burned]}\label{def:burned}
A subset $\USet$ of a set $H \subset \ZZ$ is \textbf{\burned} for $H$ if 
$$\sum_{j\in \USet}2^j\cdot
  				\uncov{\Bucket{\FirstN}{\GTSetH{\USet}\cup \USet\setminus\{j\}}}{\Bucket{\FirstN}{j}} 
  				> \doNothingMult\cdot\LOPT{\Bucket{\FirstN}{\USet}}.$$
\end{definition}

\begin{definition}\textbf{[\CTreeN]}\label{def:CriticalTree}
Let  $\mFam$ be a family of subsets of $\ZZ$ and
$H = \bigcup_{H'\in \mFam}H'$,
a \CTreeN\ for $\mFam$ is a rooted tree whose vertices are subsets of $H$ and that satisfies the following:
\begin{enumerate}
\item\label{item:TreeRoot} the root of the tree is $H$,
\item\label{item:TreeChildren} the children of the root are the sets of $\mFam$,
\item\label{item:TreeLeaves} every leaf is either \negligible, \usefulN\ or \burned, and
\item\label{item:TreeInternal} every internal vertex $\USet$, except possibly the root, is \splittable\, and neither \usefulN, \negligible\ nor \burned; moreover it has two children that form a bipartition of $\USet$ as described in the definition of  \splittable.
\end{enumerate}
\end{definition}

\begin{lemma}\label{lem:CTreeDepth}
Suppose that $\CTree$ is a \CTreeN\ for a \cf\ $\mFam$.
If $\rank{\FirstN} > 2^{\maxWeightPowerMT}$, 
then the depth of $\CTree$ does not exceed $\preDepthMult\cdot\log\log{\rank{\FirstN}}.$
\end{lemma}
\begin{proof}
Let $\USet$ be a parent of a leaf in $\CTree$ and $d$ be the depth of $K$.
We assume that $\USet$ is not the root or one of its children, since otherwise the result follows immediately.
By the definition of a \CTreeN, each ancestor $\USet^*$ of $\USet$, except for the root, is \splittable\ and hence, by the definition of \splittable, 
$\LOPT{\Bucket{\FirstN}{\USet}} \leq 
\left(\splitLBCnstOM\right)^{d-2}\cdot\LOPT{\Bucket{\FirstN}{H}}.$
Since $\USet$ is not \negligible, by definition,
$\frac{\LOPT{\Bucket{\FirstN}{H}}}{\preLOPTLBDiv\cdot\log{\rank{\FirstN}}}
 \leq
 \LOPT{\Bucket{\FirstN}{\USet}}.$ 
Therefore, 
$ 
\frac{\LOPT{\Bucket{\FirstN}{H}}}{\preLOPTLBDiv\cdot\log{\rank{\FirstN}}} \leq
 \left(\splitLBCnstOM\right)^{d-2}\cdot\LOPT{\Bucket{\FirstN}{H}}.
$
Consequently, since $\rank{\FirstN} > 2^{\maxWeightPowerMT}$, we have that
$d \leq \preDepthMult\cdot\log\log{\rank{\FirstN}}$, which in turn implies the result.
\end{proof}

\begin{lemma}\label{lem:CtreeExists}
Let $L\subset \ZZ$ and $\mathcal{H}$ be a \cf\ for $L$.
Then, there exists a \CTreeN\ $\CTree$ for $\mathcal{H}$.
\end{lemma}
\begin{proof}
We construct $\CTree$ as follows:
we let $H = \bigcup_{H'\in \mFam}H'$ and
 set the root to be $H$ and its children to be $\mathcal{H}$.
Then, as long as there is a leaf $K$ in the tree that is \splittable\ but neither
\usefulN,  \negligible, nor \burned, we add two children $K_1$ and $K_2$ to $K$, where $\{K_1,K_2\}$ form a bipartition of $K$, as in the definition of \splittable.
If there are no such leaves,  we stop.

By construction, every vertex $K$ in $\CTree$ is a subset of $H$ and $\CTree$ satisfies Items~\ref{item:TreeRoot},~\ref{item:TreeChildren} and~\ref{item:TreeInternal}  of the definition of a \CTreeN.
Suppose that every $\USet\subseteq H$, except for possibly the root, is at least one of the following: \usefulN,  \negligible, \burned\ or \splittable.
This implies that $\CTree$ also satisfies Item~\ref{item:TreeLeaves} of the definition of a \CTreeN.
Thus, $\CTree$ is a \CTreeN\ for $\mFam$.

We prove next that indeed, every $\USet\subseteq H$ is at least one of: \usefulN,  \negligible, \burned\ or \splittable.
Fix $\USet\subseteq H$ and
assume that $\USet$ is neither \burned, \negligible\ nor \splittable.
By the definition of \splittable,
there exists $\PISHeavy\in \USet$ such that,
$\LOPT{\Bucket{\FirstN}{\PISHeavy}}\geq \splitUBCnst\cdot \LOPT{\Bucket{\FirstN}{\USet}}$.
Hence, the set $\USet^*= \{\PISHeavy\}$ satisfies Item~\ref{item:usefulOPT} of the definition of \usefulN.
We show next that $\USet^*$ also satisfies Item~\ref{item:usefulUncov}
of the definition of \usefulN, and therefore is \usefulN.

Since $\USet$ is not \burned,  
$2^{\PISHeavy}\cdot \uncov{\Bucket{\FirstN}{\GTSet{\USet}\cup \USet\setminus\{\PISHeavy\}}}{\Bucket{\FirstN}{\PISHeavy}}\leq\doNothingMult\cdot \LOPT{\Bucket{\FirstN}{\USet}}$,
hence by Item~\ref{item:usefulUncov}
of the definition of \usefulN,
 it is sufficient to show that
$2^{\PISHeavy}\cdot \uncov{\Bucket{\FirstN}{\GTSetH{\{\PISHeavy\}}}}{\Bucket{\FirstN}{\PISHeavy}}\geq \UNSMultConst\cdot \LOPT{\Bucket{\FirstN}{\USet}}$.
Since $\mathcal{H}$ is a \cf,
by Item~\ref{item:CfUncov} of Definition~\ref{def:Cf} and Observation~\ref{obs:uncov},
$\uncov{\Bucket{\FirstN}{\GTSetH{\{\PISHeavy\}}}}{\Bucket{\FirstN}{j}}
\geq
\frac{\IncMultCnstHalfMO}{\IncMultCnstHalf}\cdot \rank{\Bucket{\FirstN}{\PISHeavy}}$.
Consequently, because $\LOPT{\Bucket{\FirstN}{\PISHeavy}}\geq \splitUBCnst\cdot \LOPT{\Bucket{\FirstN}{\USet}}$,
we have that 
$2^{\PISHeavy}\cdot \uncov{\Bucket{\FirstN}{\GTSetH{\{\PISHeavy\}}}}{\Bucket{\FirstN}{\PISHeavy}}
\geq
\UNSMultConst\cdot \LOPT{\Bucket{\FirstN}{\USet}}$.
\end{proof}

\begin{theorem}\label{thm:Tradeoff}
Let  $L$ be a set of at most $\valuableRangeMult\cdot\log{\rank{\FirstN}} \PValuableIndexLBAdd$ integers.
If $\rank{\FirstN} > 2^{\maxWeightPowerMT}$ and
$\sum_{j\in K}2^j\cdot
\uncov{\Bucket{\FirstN}{K\setminus\{j\}}}{\Bucket{\FirstN}{j}}\leq \frac{\LOPT{\Bucket{\FirstN}{L}}}{\PreSetCnst\cdot \log\log{\rank{\FirstN}}}$,
for every \manageable\ $K \subseteq L$, 
then there exists a \CT\ $(\BlockN,\GdN,\BdN)$ such that, for every $i\in \ZZ$,
\begin{enumerate}
\item\label{item:CTContained} $\Gd{i}$, $\Bd{i}$ and $\Block{i}$ are subsets of  $L$,
\item\label{item:CTMinIndexBucket}
$\rank{\Bucket{\FirstN}{j}} > \left(\manageableGdCnst\cdot\sum_{\ell\in \Bd{i}\cup \Block{i}} \rank{\Bucket{\FirstN}{\ell}}\right)^{\manageablePower}$, for every $j\in \Block{i}$,
 and
\item\label{item:CTValue}
$
\begin{array}{lcc}
\sum_{j\in \BLOCK}2^j\left(
\uncov{\Bucket{\FirstN}{\Bd{j}\cup \Block{j}\setminus\{j\}}}{\Bucket{\FirstN}{j}}
 - 
\uncov{\Bucket{\FirstN}{\Gd{j}\setminus\{j\}}}{\Bucket{\FirstN}{j}} \right) 
 \geq 
\frac{\LOPT{\Bucket{\FirstN}{L}}}{\PreTupleCnst\cdot \log\log{\rank{\FirstN}}}.
\end{array}
$
\end{enumerate}
\end{theorem}
\begin{proof}
By Lemma~\ref{lem:Cf}, there exists a \cf\ $\mathcal{H}$ for $L$.
Let $H = \bigcup_{H'\in \mathcal{H}}H'$.
By Lemma~\ref{lem:CtreeExists}, there exists a \CTreeN\ $\CTree$ for $\mathcal{H}$.
Let $\QS$ be the family containing all the leaves in $\CTree$.
Let $\QS_{\usefulN}$ be the family of all the sets in $\QS$ that are \usefulN\ in $H$. 
Define $\QS_{\burned}$ and $\QS_{\negligible}$ in the same manner.

We construct a tuple $(\BlockN,\GdN,\BdN)$ as follows:
for each $K\in \QS_{\usefulN}$, we pick a subset $K^*\subset K$, that is \usefulN\ for $K$, arbitrarily; then, for each $i\in K^*$, we let
$\Block{i} = K^*$, $\Bd{i} = \GTSetH{K}$ and $\Gd{i} = K\cup \GTSetH{K}$. Finally, 
for every $i$ such that $\Block{i}$ was not defined previously, we let 
$\Block{i} = \Gd{i} = \Bd{i} = \emptyset$.

By construction,
$(\BlockN,\GdN,\BdN)$ satisfies \textbf{Items~\ref{item:CTContained} and~\ref{item:CTMinIndexBucket}} of the theorem and Items~\ref{item:CTBlock},~\ref{item:CTSameBlock},~\ref{item:BlockGood} and~\ref{item:BBGStructure} of the definition of a \CT\ (Definition~\ref{def:tuple}).
By the definition of a \splittable\ set and the definition of a \CTreeN, 
for every pair of leaves $K,K'$ of $\CTree$,
either $\max{K}< \min{K'}$ or $\min{K} > \max{K'}$ and hence also $K\cap K' =\emptyset$.
Thus, by construction,
$(\BlockN,\GdN,\BdN)$ also satisfies Items~\ref{item:CTBlocks} and~\ref{item:CTLO} of the definition of a \CT.
Consequently, $(\BlockN,\GdN,\BdN)$ is a \CT.

By the construction of $(\BlockN,\GdN,\BdN)$ and the
 definition of \usefulN, to prove Item~\ref{item:CTValue} it is sufficient to show that
$\LOPT{\bigcup_{K\in \QS_{\usefulN}}\Bucket{\FirstN}{K}} \geq \portionOfUseful\cdot\LOPT{\Bucket{\FirstN}{H}}$ since,
by Item~\ref{item:CfLOPT} of Definition~\ref{def:Cf},
this implies that 
$\LOPT{\bigcup_{K\in \QS_{\usefulN}}\Bucket{\FirstN}{K}} \geq
\portionOfUsefulPostStrng\cdot\LOPT{\Bucket{\FirstN}{L}}$.
By Item~\ref{item:LOPT2} of Observation~\ref{obs:LOPT},
$\LOPT{\bigcup_{K\in \QS_{\usefulN}}\Bucket{\FirstN}{K}}$ is at least
$\LOPT{\bigcup_{K\in \QS}\Bucket{\FirstN}{K}} 
-
\LOPT{\bigcup_{K\in \QS_{\burned}}\Bucket{\FirstN}{K}} 
- 
\LOPT{\bigcup_{K\in \QS_{\negligible}}\Bucket{\FirstN}{K}}.
 $
To complete the proof, we bound each term in the preceding expression.

By the definition of a \CTreeN, $\bigcup_{K\in \QS}K = H$. 
Hence,
$\LOPT{\bigcup_{K\in \QS}\Bucket{\FirstN}{K}} = 
\LOPT{\Bucket{\FirstN}{H}}.
$
Since the sets in $\QS_{\negligible}$
are subsets of $H$, pairwise disjoint and not-empty, we see that $|\QS| \leq |H|$.
Thus, by the definition of \negligible,
$\LOPT{\bigcup_{K\in \QS_{\negligible}}\Bucket{\FirstN}{K}} \leq 
|H|\cdot \frac{\LOPT{\Bucket{\FirstN}{H}}}{\preLOPTLBDiv\cdot\log{\rank{\FirstN}}}$.
As $H\subseteq L$,
$|H|\leq |L| \leq \valuableRangeMult\cdot\log{\rank{\FirstN}}\PValuableIndexLBAdd$.
Consequently, 
$\LOPT{\bigcup_{K\in \QS_{\negligible}}\Bucket{\FirstN}{K}} \leq 
\portionOfNegligible\cdot \LOPT{\Bucket{\FirstN}{H}}.$
We next bound $\LOPT{\bigcup_{K\in \QS_{\burned}}\Bucket{\FirstN}{K}}$.

By the definition of \burned,
$\LOPT{\bigcup_{K\in \QS_{\burned}}\Bucket{\FirstN}{K}}\leq \doNothingMultInv\cdot \sum_{K\in\QS}\sum_{j\in K}2^j\cdot
  				   				\uncov{\Bucket{\FirstN}{\GTSet{K}\cup K\setminus\{j\}}}{\Bucket{\FirstN}{j}}$.
This in turn is bounded above by the sum of:
\begin{enumerate}[(a)]
\item\label{sum:One} $\sum_{H'\in \mFam}\sum_{j\in H'}2^j\cdot
  				\uncov{\Bucket{\FirstN}{\GTSet{H'}\cup H'\setminus\{j\}}}{\Bucket{\FirstN}{j}}$ and
\item\label{sum:Two} sum over every internal non-root vertex $K$, with children $K_1$ and $K_2$, of  
$\sum_{\ell=1}^{2}\sum_{j\in K_\ell}2^j\cdot
  				\uncov{\Bucket{\FirstN}{\GTSet{K_\ell}\cup K_\ell\setminus\{j\}}}{\Bucket{\FirstN}{j}} -
  				 \sum_{j\in K}2^j\cdot
  				   				\uncov{\Bucket{\FirstN}{\GTSet{K}\cup K\setminus\{j\}}}{\Bucket{\FirstN}{j}}	
  				$.
\end{enumerate}
We note that sum (\ref{sum:Two}) is the additional \uncovN\ measure because of the difference between the \uncovN\ of the children and their parent.
 
By construction, every $H'\in \mFam$ is \manageable\ and therefore,
by assumption, 
$\sum_{j\in H'}2^j\cdot
  				\uncov{\Bucket{\FirstN}{\GTSet{H'}\cup H'\setminus\{j\}}}{\Bucket{\FirstN}{j}}\leq
  				\sum_{j\in H'}2^j\cdot
  				  				\uncov{\Bucket{\FirstN}{ H'\setminus\{j\}}}{\Bucket{\FirstN}{j}}\leq \frac{\LOPT{\Bucket{\FirstN}{L}}}{\PreSetCnst\cdot \log\log{\rank{\FirstN}}}$,
  				where the first inequality follows from Observation~\ref{obs:uncov}.
Thus, the value of (\ref{sum:One}) is bounded above by
$\frac{|\mFam|\cdot\LOPT{\Bucket{\FirstN}{L}}}{\PreSetCnst\cdot \log\log{\rank{\FirstN}}}$.
Because
$|\mFam| \leq \PartitionCardCnst\cdot \log\log{\rank{\FirstN}}$,
by Lemma~\ref{lem:Cf}, we see that
$\frac{|\mFam|\cdot\LOPT{\Bucket{\FirstN}{L}}}{\PreSetCnst\cdot \log\log{\rank{\FirstN}}} \leq
\portionOfBurnedDepthOnePre\cdot \LOPT{\Bucket{\FirstN}{L}}
$.
This implies that the value of (\ref{sum:One}) is bounded above by $\portionOfBurnedDepthOne\cdot \LOPT{\Bucket{\FirstN}{H}}$,
by Item~\ref{item:CfLOPT} of Definition~\ref{def:Cf}.

By construction, every internal non-root vertex $K$ of $\CTree$ is  \splittable\ and not \usefulN.
Therefore its children do not satisfy Item~\ref{item:usefulUncov} of the definition of \usefulN. 
Hence, the value of (\ref{sum:Two}) is bounded above by the sum of 
$\frac{2\cdot\LOPT{\Bucket{\FirstN}{K}}}{\UsefulDenumCnst\cdot\log\log{\rank{\FirstN}}}$
 over every internal non-root vertex $K$ of $\CTree$.
The sum of $\LOPT{\Bucket{\FirstN}{K}}$ over all such vertices at any given depth is at most $\LOPT{\Bucket{\FirstN}{H}}$.
So, by Lemma~\ref{lem:CTreeDepth}, the value of (\ref{sum:Two}) does not exceed 
$\frac{2\cdot\preDepthMult\cdot\log\log{\rank{\FirstN}}\cdot \LOPT{\Bucket{\FirstN}{H}}}{\UsefulDenumCnst\cdot\log\log{\rank{\FirstN}}}$.
Consequently, (\ref{sum:Two}) is bounded above by
$\portionOfBurned\cdot \LOPT{\Bucket{\FirstN}{H}}$ and the result follows.
\end{proof}

\section{Main Result}\label{sec:MainResult}
The main result in this section is 
Theorem~\ref{thm:main}, which states that the \MA\ indeed has the claimed \cor.
The proof of the theorem provides the details of how the \MA\ works and utilizes Theorems~\ref{thm:greedy},~\ref{thm:tuple},~\ref{thm:Concentrations} and~\ref{thm:Tradeoff}.

One of the crucial details of the proof is that the \MA\ only involves a subset of the buckets.
Specifically, those that belong to the set $\Valuable$, defined as follows: 

\begin{definition}\textbf{[\Valuable, $L$, $L'$]}\label{def:Valuable}
  We define 
  $\Valuable = \left\{ j \Big| \rank{\Bucket{\FirstN}{j}} 
  > 
  \left(2^{-j\MValuableCnst}\cdot\LOPT{\FirstN}\right)^{\ValuablePwr}\right\},$
 $L = \{i\in \Valuable \mid i < \log{\LOPT{\FirstN}}-\valuableRankSplitMult\cdot \log\log{\rank{\FirstN}}\}$ and
$L' = \Valuable\setminus L$.  
\end{definition}

The importance of $\Valuable,L$ and $L'$, as implied by Lemma~\ref{lem:Valuable} below, is that 
Item~\ref{item:SingleBucket} of Theorem~\ref{thm:Concentrations} applies to every bucket in $L'$, and
both Theorem~\ref{thm:Concentrations} and Theorem~\ref{thm:Tradeoff} apply to the relevant subsets of $L$.
The next result, Lemma~\ref{lem:slack}, is required in order to bound the influence of the deviation in Theorem~\ref{thm:Concentrations}.

\begin{lemma}~\label{lem:Valuable}
If the event of Theorem~\ref{thm:Concentrations} holds, then
  \begin{enumerate}
  \item\label{item:ValuableCardinality}
  $|L| \leq \valuableRangeMult\cdot\log{\rank{\FirstN}} \PValuableIndexLBAdd$,  
  \item \label{item:ValuableFirstLast} for every $j\in \Super$, $\singleFirstLastRankRatio\cdot\rank{\Bucket{\LastN}{j}} \geq \rank{\Bucket{\FirstN}{j}}\geq \SuperSingleBucketFirstPortion\cdot \rank{\BucketU{j}}$,  
  \item\label{item:KOneKTwo}
  	if $\LOPT{\Bucket{\FirstN}{L'}} < \KOnePortion\cdot\LOPT{\UU}$, then
  $\LOPT{\Bucket{\FirstN}{L}} \geq \KTwoPortion\cdot\LOPT{\UU}$ and
  \item\label{item:SuperValuable} 
    $\Valuable\subseteq \Super$.
  \end{enumerate}
\end{lemma} 

\begin{proof}
We first prove Item~\ref{item:ValuableCardinality}.
Using Definition~\ref{def:LOPT},
 $\max{L} \leq \log{\LOPT{\FirstN}}$.
By Definition~\ref{def:Valuable}, 
$\left(2^{-\min{L}\MValuableCnst}\cdot\LOPT{\FirstN}\right)^{\ValuablePwr} \leq \rank{\Bucket{\FirstN}{\min{L}}} \leq \rank{\FirstN}$ and hence
$\min{L} \geq \log{\LOPT{\FirstN}} - \valuableRangeMult\cdot\log{\rank{\FirstN}}\PValuableIndexLBAddMPO.$
Since $L$ contains only integers,  \textbf{Item~\ref{item:ValuableCardinality}} follows. 

Let $j\in\Super$.
By the definition of $\Super$ (Definition~\ref{def:SuperBuckets}),
$\rank{\BucketU{j}} \geq \left(2^{-j\PSuperCnst}\cdot\LOPT{\UU}\right)^{\SuperPwr}$.
Thus, as 
$j <  \log\LOPT{\UU}-\maxWeightPower$, we see that
 $\rank{\BucketU{j}} > 2^{\singleSuperRankLB}$.
This implies that
$\TalagrandMultCnst\cdot\rank{\BucketU{j}}^{\singleUnionPowerPHalf} \leq \frac{1}{2}\cdot\rank{\BucketU{j}}$.
By Item~\ref{item:SingleBucket} of Theorem~\ref{thm:Concentrations},
$\Big|\rank{\Bucket{\LastN}{j}} - \rank{\Bucket{\FirstN}{j}}\Big| \leq 
\TalagrandMultCnst\cdot\rank{\BucketU{j}}^{\singleUnionPowerPHalf}$.
By Item~\ref{item:MatroidUnion} of Proposition~\ref{prop:Matroid}, we also have
$\rank{\Bucket{\LastN}{j}} + \rank{\Bucket{\FirstN}{j}} \geq 
\rank{\BucketU{j}}$.
The preceding three inequalities imply that  \textbf{Item~\ref{item:ValuableFirstLast}} holds.

Suppose first that
$\LOPT{\Bucket{\FirstN}{\Valuable}} \geq  \KPortionTT\cdot \LOPT{\UU}$.
Then, \textbf{Item~\ref{item:KOneKTwo}} holds and
 $\Valuable \neq \emptyset$.
Since $\LOPT{\FirstN} \geq \LOPT{\Bucket{\FirstN}{\Valuable}}$, using Definition~\ref{def:Valuable}, for every $i\in \Valuable$,
$\rank{\BucketU{i}} \geq \rank{\Bucket{\FirstN}{i}}>\left(2^{-i\PSuperCnst}\cdot\LOPT{\UU}\right)^{\SuperPwr}$.
Hence, by Definition~\ref{def:SuperBuckets}, \textbf{Item~\ref{item:SuperValuable}} follows.

Finally we prove that $\LOPT{\Bucket{\FirstN}{\Valuable}} \geq  \frac{1}{4}\cdot \LOPT{\UU}$. 
We do so by first proving that $\LOPT{\FirstN}\geq \FirstPortionInUU \cdot \LOPT{\UU}$  and then that 
$\LOPT{\Bucket{\FirstN}{\Valuable}}
  \geq \valuablePortion\cdot\LOPT{\FirstN}$.

Let $J^*$ be the set of all integers $i$ such that 
 $1 \leq \rank{\BucketU{i}} \leq \left(2^{-i\PSuperCnst}\cdot\LOPT{\UU}\right)^{\SuperPwr}$.
Thus, 
$\LOPT{\BucketU{J^*}}\leq \sum_{i\in J^*}2^i\cdot\left(2^{-i\PSuperCnst}\cdot\LOPT{\UU}\right)^{\SuperPwr}
=
\sum_{i\in J^*}\left(2^{\frac{i}{\SuperPwrEnum}\PSuperCnst}\cdot\LOPT{\UU}\right)^{\SuperPwr} 
 $. 
 This is less than $
 \sum_{\ell \geq 0}\left(2^{\frac{\max J^* - \ell}{\SuperPwrEnum}\PSuperCnst}\cdot\LOPT{\UU}\right)^{\SuperPwr} \leq
 \LOPT{\UU}\cdot\sum_{\ell> 0}2^{-\frac{\ell}{4}-\maxWeightPowerTSUPTTSP\PSuperCnstTSP} <  \frac{1}{\superPortionDenom}\cdot\LOPT{\UU}$, 
 because  $\max{J^*}
  < \log\LOPT{\UU}-\maxWeightPower$, by Assumption~\ref{ass:MaxVal}.
Now, since $\Super\cup J^*$ contains the indices of all non-empty \bucketN s, 
$\LOPT{\BucketU{\Super}}
\geq
\LOPT{\UU} - \LOPT{\BucketU{J^*}} > \superPortion\cdot\LOPT{\UU}.$ 
Thus, by Item~\ref{item:ValuableFirstLast} and Definition~\ref{def:LOPT},
$\LOPT{\FirstN} \geq \LOPT{\Bucket{\FirstN}{\Super}} \geq \SuperSingleBucketFirstPortion\cdot \LOPT{\BucketU{\Super}}
\geq \FirstPortionInUU \cdot \LOPT{\UU}.
$

Let $J$ be the set of all integers $i$ such that 
 $1 \leq \rank{\Bucket{\FirstN}{i}} \leq \left(2^{-i\MValuableCnst}\cdot\LOPT{\FirstN}\right)^{\ValuablePwr}$.
Thus, 
$\LOPT{\Bucket{\FirstN}{J}}\leq \sum_{i\in J}2^i\cdot\left(2^{-i\MValuableCnst}\cdot\LOPT{\FirstN}\right)^{\ValuablePwr}=\sum_{i\in J}\left(2^{\frac{i}{\ValuablePwrEnum}\MValuableCnst}\cdot\LOPT{\FirstN}\right)^{\ValuablePwr} 
 $. 
 This is less than 
 $\sum_{\ell\in J}
 \left(2^{\frac{\max{J}-\ell}{\ValuablePwrEnum}\MValuableCnst}\cdot\LOPT{\FirstN}\right)^{\ValuablePwr}
= 
 \LOPT{\FirstN}\cdot\sum_{\ell>0}2^{-\frac{\ell}{\ValuablePwrDenom}-\maxWeightPowerMTTSUPTTSP\MValuableCnstTVP} <  \frac{1}{\valuablePortionDenom}\cdot\LOPT{\FirstN}$, 
 because $\max{J} <  \log\LOPT{\FirstN}-\maxWeightPower$, by Assumption~\ref{ass:MaxVal}.
Consequently, since $\Valuable\cup J$ contains all the indices of non-empty buckets,  
  $\LOPT{\Bucket{\FirstN}{\Valuable}}
  \geq
  \LOPT{\FirstN} - \LOPT{\Bucket{\FirstN}{J}} > \valuablePortion\cdot\LOPT{\FirstN}.$ 

\end{proof}

\begin{lemma}\label{lem:slack}
Suppose that $\rank{\FirstN} >  2^{\maxWeightPowerMT}$.
Let $K\subseteq L$
and  $(\BlockN,\GdN,\BdN)$ be a \CT.
If
Item~\ref{item:ValuableFirstLast} of Lemma~\ref{lem:Valuable} holds,
 $K$ is \manageable\ and $(\BlockN,\GdN,\BdN)$
 satisfies Items~\ref{item:CTContained} and~\ref{item:CTMinIndexBucket} of Theorem~\ref{thm:Tradeoff} then,
\begin{enumerate}
\item\label{item:SlackManagable}
 $
 \TalagrandMultCnstTT\cdot\rank{\BucketU{K}}^{\simpleGapRankExp}\cdot\sum_{i\in K}
  2^i
 < 
\frac{\LOPT{\FirstN}}{\SlackCnst\cdot \log\log{\rank{\FirstN}}}.$
\item\label{item:SlackTuple}
 $
 \TalagrandMultCnstTTPAzzumaDevMultCnstTT\cdot\sum_{i\in \BLOCK}  2^i\cdot\rank{\BucketU{\Bd{i}\cup \Block{i}}}^{\simpleGapRankExp}
 < 
\frac{\LOPT{\FirstN}}{\SlackCnst\cdot \log\log{\rank{\FirstN}}}.$
\end{enumerate}
\end{lemma}
\begin{proof}
Let $c^{-1} = \min\{\rank{\Bucket{\FirstN}{i}}\mid i\in K\cup L\}^{\slackPwr}$.
By the properties of Matroids, Definition~\ref{def:manageable} and
 Item~\ref{item:ValuableFirstLast} of Lemma~\ref{lem:Valuable},
\begin{equation}\label{equ:SlackOne}
 \begin{array}{rlc}
 \TalagrandMultCnstTT\cdot\rank{\BucketU{K}}^{\simpleGapRankExp}\cdot\sum_{i\in K} 2^i 
& \leq 
 \TalagrandMultCnstTT\cdot\sum_{i\in K} 2^i\cdot\left(\sum_{j\in K}\rank{\BucketU{j}}\right)^{\simpleGapRankExp}\\
& \leq 
 \TalagrandMultCnstTT\cdot \sum_{i\in K}2^i\cdot\left(\SuperSingleBucketFirstPortionInv\cdot\sum_{j\in K} \rank{\Bucket{\FirstN}{j}}\right)^{\simpleGapRankExp} \\
& <
 \slackMultCnst\cdot\sum_{i\in K}2^i\cdot\rank{\Bucket{\FirstN}{i}}^{\slackPwrOM} \\
& \leq
 \slackMultCnst\cdot
  \min\{\rank{\Bucket{\FirstN}{i}}\mid i\in K\}^{-\slackPwr}
 \cdot\sum_{i\in K}2^i\cdot\rank{\Bucket{\FirstN}{i}} \\
& \leq
 \slackMultCnst\cdot c\cdot\LOPT{\FirstN}
  \end{array}
\end{equation}
By the properties of Matroids,  Item~\ref{item:ValuableFirstLast} of Lemma~\ref{lem:Valuable}, and Items~\ref{item:CTContained} and~\ref{item:CTMinIndexBucket} of Theorem~\ref{thm:Tradeoff},
\begin{equation}\label{equ:SlackTwo}
 \begin{array}{rlc}
 \TalagrandMultCnstTTPAzzumaDevMultCnstTT\cdot\sum_{i\in \BLOCK}  2^i\cdot\rank{\BucketU{\Bd{i}\cup \Block{i}}}^{\simpleGapRankExp}  
& \leq  
\TalagrandMultCnstTTPAzzumaDevMultCnstTT\cdot\sum_{i\in \BLOCK}  2^i\cdot
 \left(\SuperSingleBucketFirstPortionInv\cdot\sum_{j\in\Bd{i}\cup \Block{i}}\rank{\Bucket{\FirstN}{j}}\right)^{\simpleGapRankExp} \\
& <
  \slackMultCnst\cdot\sum_{i\in \BLOCK}  2^i\cdot\rank{\Bucket{\FirstN}{i}}^{\slackPwrOM} \\
&   \leq 
   \slackMultCnst\cdot c\cdot\LOPT{\FirstN}.
 \end{array}
\end{equation}
\noindent
By Definition~\ref{def:Valuable},  
$$c^{-1} \geq \left(2^{-\log{\LOPT{\FirstN}}+\valuableRankSplitMult\cdot \log\log{\rank{\FirstN}}\MValuableCnst}\cdot\LOPT{\FirstN}\right)^{\slackPwrTVal}
>
\frac{1}{2}\cdot\log{\rank{\FirstN}}^{\slackLogExpCnst}\cdot\log{\log{\rank{\FirstN}}},$$
so,
$c^{-1} > \SlackCnstTTF\cdot\log{\log{\rank{\FirstN}}},$ because
$\rank{\FirstN} > 2^{\maxWeightPowerMT}$.
Thus, by (\ref{equ:SlackOne}) and (\ref{equ:SlackTwo}), the result follows.
\end{proof}

The following theorem is the main result of this paper.
\begin{theorem}~\label{thm:main}
The \MA\ is \OO, \OC\ and has returns and, with constant probability, returns and independent set of elements of value $\Omega(\frac{\OPT{\UU}}{\log{\log{\rank{\UU}}}})$.
\end{theorem}
\begin{proof}
We note that the \MA\ is \OC, 
since the computation in \StageOne\ is independent of the matroid elements and 
the computation in the \StageTwoWS\ and the \StageThreeWS\ stages uses only elements of the matroid that have already been revealed.
We also note that the \MA\ is \OO, because by construction, and following Definition~\ref{def:OO}, the analysis depends on the elements in the sets $\FirstN$ and $\LastN$ but not on their order.

By Assumption~\ref{ass:MaxVal},
 the properties of Matroids and Observation~\ref{obs:OPT}, it follows that $\rank{\UU} > 2^{\maxWeightPower}$.
Thus, the event in Theorem~\ref{thm:Concentrations} holds
with probability at least $\TradeoffProb$.
So, by the definition of \cor, it is sufficient to prove the result assuming the event in Theorem~\ref{thm:Concentrations} holds.
We proceed on this assumption.
We note that this means that the conditions needed for Lemma~\ref{lem:Valuable} hold.
By Item~\ref{item:rankFirst} of Theorem~\ref{thm:Concentrations},
we also have $\rank{\FirstN} > 2^{\maxWeightPowerMT}$.
Thus, the conditions needed for Lemma~\ref{lem:slack} also hold.

To conclude the proof we require the use of Items~\ref{item:ConcentrationUncov} and~\ref{item:ConcentrationIndep} of Theorem~\ref{thm:Concentrations}.
We now, prove that they hold for the sets relevant to the proof.
By Item~\ref{item:rankFirst} of 
 Theorem~\ref{thm:Concentrations} 
and Definition~\ref{def:Valuable}, 
we see that $\max{L} \leq \log\LOPT{\UU} - \SuperSplitMultCnst\cdot \log\log{\rank{\UU}}$.
Also, by Item~\ref{item:SuperValuable} of Lemma~\ref{lem:Valuable} and Definition~\ref{def:Valuable}, $L\subseteq \Valuable \subseteq \Super$.
Hence, 
 Items~\ref{item:ConcentrationUncov} and~\ref{item:ConcentrationIndep} of Theorem~\ref{thm:Concentrations} hold, for every $K,K'\subseteq L$, where  $\min{K'} > \max{K}$ or $K' =\emptyset$ and 
 $\min\{\rank{\BucketU{j}}\mid j\in K\} \geq \left(\concentrationRangeMultCnst\cdot\rank{\BucketU{\min{K}}}\right)^{\manageablePower}$, and
 for every $k\in K$.

\paragraph{Case 1:}
Suppose that $\LOPT{\Bucket{\FirstN}{L'}} \geq \KOnePortion\cdot \LOPT{\UU}$.
We observe that $|L'| \leq \valuableRankSplitMult\cdot \log\log{\rank{\FirstN}}$ since, by Definition~\ref{def:LOPT}, $\max{\Valuable}\leq \log{\LOPT{\FirstN}}$. 
Therefore, by the Pigeon Hole Principle and  Definition~\ref{def:LOPT}, 
there exists $k\in L'$ such that 
$\LOPT{\Bucket{\FirstN}{k}} \geq \frac{\LOPT{\UU}}{\LastSetCnst\cdot \log\log{\rank{\FirstN}}}$.
By Items~\ref{item:ValuableFirstLast} and~\ref{item:SuperValuable} of Lemma~\ref{lem:Valuable} and Corollary~\ref{cor:SA}, on input $J=\{k\}$, the \SA\ will return an independent set of elements with an optimal value of at least 
$\frac{\LOPT{\UU}}{\LastSetCnstTT\cdot \log\log{\rank{\FirstN}}}$.
Since $\rank{\FirstN} < \rank{\UU}$, it follows that this is $\Omega(\frac{\OPT{\UU}}{\log{\log{\rank{\UU}}}})$.

\paragraph{Case 2:} Suppose that $\LOPT{\Bucket{\FirstN}{L'}} < \KOnePortion \cdot \LOPT{\UU}$ and  there exists a \manageable\ set $J\subseteq L$ such that $\sum_{j\in J}2^j\cdot
\uncov{\Bucket{\FirstN}{J\setminus\{j\}}}{\Bucket{\FirstN}{j}} \geq\frac{\LOPT{\Bucket{\FirstN}{L}}}{\PreSetCnst\cdot \log\log{\rank{\FirstN}}}$.
By  
Item~\ref{item:KOneKTwo} of Lemma~\ref{lem:Valuable},
$\LOPT{\Bucket{\FirstN}{L}} 
\geq 
\KOnePortion\cdot\LOPT{\UU}$.
By Definition~\ref{def:manageable} and Items~\ref{item:ValuableFirstLast} and~\ref{item:SuperValuable} of Lemma~\ref{lem:Valuable},
for every $j\in J$,
$$\rank{\BucketU{j}} \geq \rank{\Bucket{\FirstN}{j}} \geq \left(\manageableMultCnst\cdot\sum_{i\in J} \rank{\Bucket{\FirstN}{i}}\right)^{\manageablePower} \geq
\left(\SetMainMultCnst\cdot\rank{\BucketU{\min{J}}}\right)^{\manageablePower}.$$
Thus, it follows that   
 Item~\ref{item:ConcentrationUncov} of Theorem~\ref{thm:Concentrations} holds 
  with $K=J$ and $K' = \emptyset$.
So, by Theorem~\ref{thm:greedy}, on input $J$, the \SA\ will return an independent set of elements whose optimal value is at least 
$\frac{\LOPT{\Bucket{\FirstN}{L}}}{\PreSetCnst\cdot \log\log{\rank{\FirstN}}} -
\TalagrandMultCnstTT\cdot\sum_{j\in J}2^j\cdot \rank{\BucketU{J}}^{\simpleGapRankExp}.
$
Since $\rank{\FirstN} < \rank{\UU}$ and $\LOPT{\Bucket{\FirstN}{L}} 
\geq 
\KOnePortion\cdot\LOPT{\UU}$ and $J$ is \manageable,
using Item~\ref{item:SlackManagable} of Lemma~\ref{lem:slack},
the preceding value is $\Omega\left(\frac{\OPT{\UU}}{\log\log{\rank{\UU}}}\right)$. 

\paragraph{Case 3:} Suppose that $\LOPT{\Bucket{\FirstN}{L'}} < \KOnePortion \cdot \LOPT{\UU}$ and that the assumption that the \manageable\ set $J$ exists does not hold.
By Item~\ref{item:ValuableCardinality} of Lemma~\ref{lem:Valuable}, Theorems~\ref{thm:Tradeoff} holds.
Hence, there exists a \CT\ $(\BlockN,\GdN,\BdN)$ as described in  Theorem~\ref{thm:Tradeoff}, which specifically satisfies:
\begin{equation}\label{equ:MainSum}
\sum_{j\in \BLOCK}2^j\cdot\left(
\uncov{\Bucket{\FirstN}{\Bd{j}\cup \Block{j}\setminus\{j\}}}{\Bucket{\FirstN}{j}}
 - 
\uncov{\Bucket{\FirstN}{\Gd{j}\setminus\{j\}}}{\Bucket{\FirstN}{j}} \right) 
\geq 
\frac{\LOPT{\Bucket{\FirstN}{L}}}{\PreTupleCnst\cdot \log\log{\rank{\FirstN}}}.
\end{equation}
By Items~\ref{item:CTContained} and~\ref{item:CTMinIndexBucket} of Theorem~\ref{thm:Tradeoff} and Item~\ref{item:ValuableFirstLast} of Lemma~\ref{lem:Valuable},
for every $i\in \BLOCK$ and $j\in\Block{i}$, we have that $\Gd{i}$, $\Bd{i}$ and $\Block{i}$ are subsets of $L$ and
$$\rank{\BucketU{j}} \geq \rank{\Bucket{\FirstN}{j}}
 > 
 \left(\manageableGdCnst\cdot\sum_{\ell\in \Bd{i}\cup \Block{i}} \rank{\Bucket{\FirstN}{\ell}}\right)^{\manageablePower}
 \geq 
 \left(\tupleMainMultCnst\cdot \rank{\BucketU{\min{\Block{i}}}}\right)^{\manageablePower}.
 $$
Hence, using Definition~\ref{def:tuple}, Theorem~\ref{thm:tuple} and
Items~\ref{item:ConcentrationUncov} and~\ref{item:ConcentrationIndep} of Theorem~\ref{thm:Concentrations}, we get that
$$
 \Big|\uncov{\Bucket{\FirstN}{\Bd{i}}\cup \Bucket{\LastN}{\Block{i}\setminus\{i\}}}{\Bucket{\LastN}{i}} 
-
\uncov{\Bucket{\FirstN}{\Bd{i}\cup\Block{i}\setminus\{i\}}}{\Bucket{\FirstN}{i}}\Big| \leq
\TalagrandMultCnstTT\rank{\BucketU{\Bd{i}\cup \Block{i}}}^{\simpleGapRankExp},
$$
and
$$\indep{\Bucket{\FirstN}{\Gd{i}}}{\Bucket{\LastN}{i}}
 \leq
    \uncov{\Bucket{\FirstN}{\Gd{i}\setminus\{i\}}}{\Bucket{\FirstN}{i}} + \AzzumaDevMultCnstTT\cdot \rank{\BucketU{i}}^{\rZZDiff}.$$
So, by Theorem~\ref{thm:tuple}, 
using (\ref{equ:MainSum}) it is straightforward to show that
given
$(\BlockN,\GdN,\BdN)$ as input, the
 \GAPA\ returns an independent set of elements whose optimal value is at least 
$$\frac{\LOPT{\Bucket{\FirstN}{L}}}{\PreTupleCnst\cdot \log\log{\rank{\FirstN}}}
-
\sum_{j\in \BLOCK}2^j\cdot 
\left(
\AzzumaDevMultCnstTT\cdot\rank{\BucketU{j}}^{\simpleGapRankExp} +
\TalagrandMultCnstTT\cdot\rank{\BucketU{\Bd{j}\cup\Block{j}}}^{\simpleGapRankExp}
\right).
$$
This, in turn, is bounded below by
$ 
\frac{\LOPT{\Bucket{\FirstN}{L}}}{\PreTupleCnst\cdot \log\log{\rank{\FirstN}}}
-
\TalagrandMultCnstTTPAzzumaDevMultCnstTT\cdot
\sum_{j\in \BLOCK}2^j\cdot 
\rank{\BucketU{\Bd{j}\cup\Block{j}}}^{\simpleGapRankExp}.
$
Since $\rank{\FirstN} < \rank{\UU}$ and $\LOPT{\Bucket{\FirstN}{L}} 
\geq 
\KOnePortion\cdot\LOPT{\UU}$,
using Item~\ref{item:SlackTuple} of Lemma~\ref{lem:slack},
the preceding value is $\Omega\left(\frac{\OPT{\UU}}{\log\log{\rank{\UU}}}\right)$. 

Recall that, if Assumption~\ref{ass:MaxVal} does not hold, then the \TA\ ensures that with constant probability an independent set of one element of value $\Omega\left(\frac{\OPT{\UU}}{\log\log{\rank{\UU}}}\right)$ is returned. 
Hence, we may assume that  Case 1, Case 2 or Case 3 holds.
So, in \StageTwo, we can check, using only the knowledge obtained about the elements of $\FirstN$ via the oracle, which one of the cases hold as follows:
First compute $\rank{\FirstN}$. 
Then, use $\rank{\FirstN}$ to determine the sets $\Valuable,L$ and $L'$.
Now find every \manageable\ subset of $L$ and every \CT\ that satisfies the items of Theorem~\ref{thm:Tradeoff}.
Using this information check if there exists a bucket as guaranteed if Case 1 holds, a \manageable\ set as guaranteed if Case 2 holds, or a \CT\ as guaranteed if Case 3 holds.
The analysis of the cases ensures that at least one of the preceding exists. Pick arbitrarily if there exists more than one option.
Finally, in the case of a single \bucketN\ or a \manageable\ set proceed to \StageThree\ and use the \SA, otherwise proceed to \StageThree\ and use the \GAPA.
\end{proof}

\section{Discussion}\label{sec:Discussion}
The \MA\ achieves only the claimed \cor, when the following hold:
the maximum value of an element of the Matroid is  $O\left(\frac{\OPT{\UU}}{\log{\log{\rank{\UU}}}}\right)$ and, with probability at least $1-O(\log\log{\rank{\UU}}^{-1})$,
\begin{enumerate}
\item\label{item:singleBadCase}
$2^j\cdot \rank{\Bucket{\FirstN}{j}} =O\left(\frac{\LOPT{\UU}}{\log{\log{\rank{\UU}}}}\right)$, 
for every $j\in \Valuable$,
\item\label{item:SetBadCase}
 $\sum_{j\in J} 2^j\cdot \uncov{\Bucket{\FirstN}{J\setminus \{j\}}}{\Bucket{\FirstN}{j}} = O\left(\frac{\LOPT{\UU}}{\log{\log{\rank{\UU}}}}\right)$,
for every \manageable\ subset $J$ of the set $L$ used in Theorem~\ref{thm:main}, and
\item\label{item:TupleBadCase} for every \CT\ $(\BlockN,\GdN,\BdN)$ in $L$, that satisfies the items of Theorem~\ref{thm:Tradeoff}\newline
$\sum_{j\in \ZZ}2^j\cdot\left(
\uncov{\Bucket{\FirstN}{\Bd{j}\cup \Block{j}\setminus\{j\}}}{\Bucket{\FirstN}{j}}
 - 
\uncov{\Bucket{\FirstN}{\Gd{j}\setminus\{j\}}}{\Bucket{\FirstN}{j}} \right) = O\left(\frac{\LOPT{\UU}}{\log{\log{\rank{\UU}}}}\right).$
\end{enumerate}
Understanding this case may lead to improved algorithms for the problem or, conversely, to non-trivial lower bounds.
Clearly, the \GAPA\ will perform better than claimed when 
the maximum value of an element of the Matroid is significantly larger than $\frac{\OPT{\UU}}{\log{\log{\rank{\UU}}}}$.
The above also implies that the \GAPA\ will perform better in many other cases, for example, if one of the following occurs with constant probability:
\begin{enumerate}
\item There exists $j\in \Valuable$ such that
$2^j\cdot \rank{\Bucket{\FirstN}{j}} >> \frac{\LOPT{\UU}}{\log{\log{\rank{\UU}}}}$,
which implies that Item~\ref{item:singleBadCase} above does not hold.
\item 
There exists a \manageable\ $J\subseteq L$, such that 
$\sum_{j\in J} 2^j\cdot \rank{\Bucket{\FirstN}{j}} >> \frac{\LOPT{\UU}}{\log{\log{\rank{\UU}}}}$ and
$\BucketU{J}$ is an independent set, which implies that Item~\ref{item:SetBadCase} above does not hold.
\item 
The sum of $2^j\cdot \rank{\Bucket{\FirstN}{j}}$
over every even $j\in L$, is significantly larger than
$\frac{\LOPT{\UU}}{\log{\log{\rank{\UU}}}}$, and for every even $j$ in the set $L$, that is used in Theorem~\ref{thm:main},
$\rank{\Bucket{\FirstN}{i+1}}< 4\cdot \rank{\Bucket{\FirstN}{i}}$ and
$\uncov{\Bucket{\FirstN}{i+1}}{\Bucket{\FirstN}{i}} = 0$.
This implies that Item~\ref{item:TupleBadCase} above does not hold.
\end{enumerate}

\bibliographystyle{plain}

\thispagestyle{empty}
\end{document}